\documentclass[reprint,amsmath,amssymb,aps,nofootinbib]{revtex4-2}

\usepackage{graphicx}
\usepackage{dcolumn}
\usepackage{bm}

\usepackage[utf8]{inputenc}
\usepackage[english]{babel}
\usepackage[T1]{fontenc}
\usepackage{newunicodechar}
\newunicodechar{̧}{\c{}}
\usepackage{mathrsfs}
\usepackage{amsmath}
\usepackage{amssymb}
\usepackage{amsthm}
\usepackage{balance}
\usepackage[hidelinks]{hyperref}
\usepackage{lipsum}
\usepackage{xcolor}
\usepackage{tikz}
\usetikzlibrary{graphs}

\DeclareMathOperator{\Sp}{sp}
\DeclareMathOperator{\Tr}{Tr}
\DeclareMathOperator{\Span}{Span}
\DeclareMathOperator{\Deg}{deg}
\DeclareMathOperator{\Diam}{diam}
\DeclareMathOperator{\alg}{Alg}
\DeclareMathOperator{\rk}{rank}

\newcommand{\ddt}[1]{\frac{d #1}{dt}}
\newtheorem{theorem}{Theorem}[section]
\newtheorem{lemma}{Lemma}[section]
\newtheorem{corollary}{Corollary}[section]

\newtheorem{definition}{Definition}[section]

\preprint{APS/123-QED}

\begin{document}

\title{Asymptotic relaxation in quantum Markovian dynamics}

\author{Giovanni Di Meglio$^{1,2}$}
\email{giovanni.di-meglio@uni-ulm.de}
\author{Dariusz Chruściński$^{3}$}
\email{darch@fizyka.umk.pl}
\author{Koenraad Audenaert$^{1}$}
\email{koenraad.audenaert@uni-ulm.de}
\author{Martin B. Plenio$^{1,2}$}
\email{martin.plenio@uni-ulm.de}
\author{Susana F. Huelga$^{1,2}$}
\email{susana.huelga@uni-ulm.de}
 \affiliation{$^1$Institut für Theoretische Physik,
Albert-Einstein-Allee 11, Universität Ulm, D-89081 Ulm, Germany
}
\affiliation{$^2$ Center for Integrated Quantum Science and Technology (IQST), Ulm University, Albert-Einstein-Allee 11, 89081 Ulm, Germany}

\affiliation{$^3$Institute of Physics, Faculty of Physics, Astronomy and Informatics, Nicolaus Copernicus University, Grudzia̧dzka 5/7, 87-100 Toruń, Poland}

\date{\today}

             
\begin{abstract}
We investigate the long-time behavior of quantum Markovian dynamics generated by time-dependent Gorini–Kossakowski–Lindblad–Sudarshan (GKLS) master equations. 
We introduce a notion of weak relaxation and derive sufficient conditions guaranteeing asymptotic independence from the initial state. 
Our results provide a quantitative extension of the Spohn–Frigerio theorem to time-dependent generators, yielding explicit contraction bounds in terms of the instantaneous 
steady state and time-integrated dissipation rates. For a class of microscopically derived master equations, we further obtain a graph-theoretic characterization of the aforementioned conditions 
that directly links the structure of the jump operators to the relaxation properties. The general theory is illustrated by applications to driven finite-level systems, 
including a detailed three-level example, and is extended to a non-Markovian setting by means of time-local master equations that become of GKLS form at long times.
These findings pave the way for the development of a more general theory of relaxation beyond the Markovian case.
\end{abstract}

\maketitle

\section{\label{sec:intro}Introduction}

Among the different methods employed to study the behavior of open quantum systems, the analysis of dynamical equations of the system density matrix, and in particular the celebrated Gorini-Kossakowski-Lindblad-Sudarshan (GKLS) master equation, is widely used (\cite{Lindblad, Gorini, Breuer, Rivas}).
Since the early works on the topic, the investigation of the asymptotic behavior of such dissipative quantum systems has attracted considerable interest.  On the one hand, understanding the emergence of equilibrium and non-equilibrium stationary states represents a central problem in the dynamical description of large quantum systems, starting from the underlying microscopic model \cite{SpohnReview}. On the other hand, the possibility to suitably engineer reservoirs has been shown to be an efficient tool to drive quantum systems towards desired steady states, as for instance a highly entangled state, which can be used as a resource for useful quantum information tasks (\cite{PH_1,PH_2,Diehl,Kraus, Verstraete,Ticozzi, Schirmer}).

In the case of time-independent Lindbladians, a rich literature has developed results concerning uniqueness and characterization of the steady state of the dynamics \cite{Davies, Spohn_1,Spohn_2,Evans,Frigerio_1,Frigerio_2,Baumgartner_1,Baumgartner_2,Schirmer, Albert, ZhangZ, Yoshida, ZhangY}, whereas very little is known in the time-dependent case. Therefore, the theoretical investigation of the long-time behavior of general, time-dependent, open quantum systems remains a demanding problem.

Situations where the description of the open dynamics requires the use of time-dependent GKLS equations are ubiquitous in the presence of driving, such as those in quantum control theory \cite{Koch} or in quantum thermodynamics \cite{Binder}.

Currently, known results  are mainly confined to the important case of periodic driving and analysis of specific models \cite{Feldmann_1,Insinga_0,Insinga_1,Menczel, Ikeda2020,Ikeda2021}.


In order to fill this gap, the purpose of the present study is to take a step forward, providing an alternative methodology to address the independence on the initial state of the long-time dynamics for arbitrary driving.

First, we properly define a generalized form of relaxation (\emph{weak relaxation}), which does not require the existence of a well-defined state for $t\rightarrow \infty$, in contrast with the standard definition used in the context of quantum dynamical semigroups (\emph{strong relaxation}).

Second, our main results, stated in Theorem \ref{th: main_1} and \ref{th: main_2}, address the long-time dynamics by means of asymptotic analysis.

While the main theorems of this work are completely general, we consider in Section \ref{sec:appl} a specific type of GKLS master equation for driven systems, recently derived in \cite{DiMeglio2024}. In this case, we provide a direct link between the microscopic dynamics and the quantities responsible for relaxation, such as the smallest dissipation rate and the spectrum of the instantaneous steady state of the dynamics. In order to address the sufficient conditions for relaxation, we introduce a graph-theoretic construction associated with the set of jump operators and their specific structure, motivated by the microscopic Hamiltonian. Utilizing known results in the theory of graphs, we relate the algebraic conditions for relaxation with the connectivity of such graph. 

Finally, our theory is concretely applied to a driven 3-level system, which allows us to show the feasibility  of our methodology.

This paper is organized as follows. First, we expose in Section \ref{sec:summary} the general theory, whose novel results are summarized in Theorem \ref{th: main_1} and \ref{th: main_2}. 
Section \ref{sec:proof} is entirely devoted to the proof of the theorems, whereas Section \ref{sec:appl} treats in detail a specific GKLS master equation derived from a microscopic Hamiltonian. Finally, we will draw our conclusions and offer future perspectives on the present topic in Section \ref{sec:concl}. Further details can be found also in the appendices, including a complementary analysis of the relaxation in the Heisenberg picture (Appendix \ref{sec:app_Heisenberg}) and a self-contained introduction to the results concerning spectral graph theory employed in this paper (Appendix \ref{app: graph}).

\section{\label{sec:summary}General theory}

\subsection{\label{sec: preliminaries}Preliminaries and notation}
Let $\mathscr{H}$ be the Hilbert space of pure states with $\dim\mathscr{H}=N$ and $\mathscr{B}$ the set of linear operators acting on $\mathscr{H}$. The identity operators acting on $\mathscr{H}$ and $\mathscr{B}$ will be indicated, respectively, by $1_H$ and $1_B$. The spectrum of an operator $a$ is denoted by $\Sp(a)$, moreover it is convenient in our treatment to introduce the notation $m(a)=\inf \Sp(a)$ and $M(a)=\sup \Sp(a)$. The space $\mathscr{B}$ can be equipped with different norms, in particular we will make use of the relevant class of Schatten $p$-norms for $p=1,2,\infty$, namely $||a||_1=\Tr\sqrt{a^{\dagger}a}$, $||a||_2=\sqrt{\Tr(a^{\dagger}a)}$ and $||a||_{\infty}=\sup \Sp(\sqrt{a^{\dagger}a})$, where $a^{\dagger}$ denotes the adjoint element of $a$, with respect to the scalar product on $\mathscr{H}$.  
The subspace of self-adjoint elements $a=a^{\dagger}$ is indicated with $\mathscr{S}$, whereas $\mathscr{S}^+$ will denote the set of positive elements $a\ge 0$, i.e. positive semi-definite operators, $\mathscr{S}^+_1$ the set of quantum states, i.e. positive semi-definite operators with trace one, and  $\mathscr{S}_0$ the space of self-adjoint traceless operators.
Moreover, we will make use of the Hilbert-Schmidt scalar product $(a,b)_{HS}=\Tr(a^{\dagger}b)$ on  $\mathscr{B}$, with $||a||_2^2=(a,a)_{HS}$. Hence, unless differently specified, the notion of orthogonality of operators must be understood with respect to this scalar product.
For linear maps $T: \mathscr{B}\rightarrow \mathscr{B}$, we consider the submultiplicative norms induced by the Schatten p-norms, namely
\begin{equation}
||T||_{p,p}=\sup\limits_{a\neq 0}\frac{||T(a)||_p}{||a||_p},
\end{equation}
for $p=1,2,\infty$. Furthermore, due to the Hilbert-Schmidt scalar product, a notion of adjoint elements can be introduced for a superoperator $T$ as well, which we denote with $T^{\ddagger}$.

This completes our preliminary notation. 

\begin{definition}
A quantum evolution map or dynamical map $\{\Lambda_{t,0}\}_{t\ge 0}$ is a family of CPTP (completely positive, trace preserving) and Hermiticity preserving linear operators on $\mathscr{B}$.
\end{definition}

Throughout this paper we consider an important subclass of legitimate evolution maps, namely Markovian quantum processes \cite{Lindblad, Gorini, Rivas}. 

\begin{definition}
A quantum Markovian process is a dynamical map $\Big\{ \Lambda_{t,0}=\mathcal{T}\exp\int\limits_{0}^{t}ds\mathcal{L}_s \Big\}_{t\ge 0}$ generated by the time-dependent GKLS generator 
\begin{equation}
\begin{aligned}
\label{markov_eq}
\mathcal{L}_t\rho=& -\imath[H(t),\rho]+\sum_{\alpha}\gamma_{\alpha}(t)\\
&\times\Big(L_{\alpha}(t)\rho L_{\alpha}^{\dagger}(t)-\frac{1}{2}\Big\{L_{\alpha}^{\dagger}(t)L_{\alpha}(t),\rho \Big\} \Big),
\end{aligned}
\end{equation} 
with $\gamma_{\alpha}(t)\ge 0, \forall\alpha,\forall t>0$, $H(t)$ self-adjoint. 
\end{definition}
In this work, the notion of Markovianity corresponds to CP-divisibility of the dynamical map, i.e. the intermediate maps $\Lambda_{t_2,t_1}$ are CP for any $t_2\ge t_1\ge 0$ (\cite{Wolf, Rivas_NM}).

The purpose of the present work is to address the long-time behavior of the density matrix $\rho(t)=\Lambda_{t,0}\rho(0)$ for arbitrary initial states $\rho(0)\in\mathscr{S}_1^+$ and in particular to determine under which conditions the dynamics for large times is independent of the initial state. This so-called relaxation property has been explicitly defined in the context of quantum dynamical semigroups (see for example \cite{Spohn_1,Spohn_2,Frigerio_1}) and it is related, in the finite dimensional case, to the uniqueness of the steady state, namely  $\dim(\ker\mathcal{L})=1$ (\cite{Schirmer}). We can state it for arbitrary dynamical maps as
\begin{definition}
\label{def: strong}
Let $\{\Lambda_{t,0}\}_{t\ge 0}$ be a quantum dynamical map, we say that the dynamics is strongly relaxing iff there exists a unique $ \rho_{st}\in\mathscr{S}_1^+$ such that, $\forall\rho\in\mathscr{S}_1^+$, $\lim\limits_{t\rightarrow\infty}||\Lambda_{t,0}(\rho) - \rho_{st}||_1=0$.
\end{definition}

However, this definition is less natural in the presence of driving and control fields, which can prevent the system from reaching a well-defined stationary state. Therefore, the concept of relaxation must be suitably extended.

\begin{definition}
\label{def: weak}
Let $\{\Lambda_{t,0}\}_{t\ge 0}$ be a quantum dynamical map, we say that the dynamics is weakly relaxing iff there exists a trajectory $t\in \mathbb{R}^+\rightarrow \theta(t) \in\mathscr{S}_1^+$ such that, $\forall\rho\in\mathscr{S}_1^+$, $\lim\limits_{t\rightarrow\infty}||\Lambda_{t,0}(\rho) - \theta(t)||_1=0$.
\end{definition}

We underline that, while the convergence is stated in terms of the $p=1$ norm, in principle any other norm can be equally used, since we are working in finite dimensional spaces. Clearly, Definition \ref{def: strong} is a special case of Definition \ref{def: weak}. In particular, in the time-independent case the trajectory $\theta(t)$ can be identified, at least asymptotically, with the fixed point $\omega\in \ker\mathcal{L}$ of the dynamics, which is a steady state.

In the time-dependent case, however, the situation is more complicated.
Under the hypotheses of the theory that we will develop below, for which the kernel of the time-dependent generator $\mathcal{L}_t$ is at any time $t$ one-dimensional, we can
define a second type of state trajectory $\omega(t)$. For any fixed value $t=t_0$, $\omega(t_0)$ is given by the unique steady state 
of the  map $e^{\mathcal{L}_{t_0}}$, obtained from the fixed-time generator $\mathcal{L}_{t_0}$.
In the literature, $\omega(t)$ is often called the \emph{instantaneous steady state}, but one should 
not misconstrue this as implying that the trajectory
$\omega(t)$ has anything to do with the limit trajectory $\theta(t)$, or that $\omega(t)$ is the steady state of the time-dependent map $\Lambda_{t,0}$.
Nevertheless, this second state trajectory $\omega(t)$ will play a central role in our theory.


The existence of at least one non-trivial element in $\ker\mathcal{L}_t$, as well as the fact that it can be taken as positive semidefinite, represent general results in the theory of quantum dynamical semigroups (\cite{Rivas}). 
In this work, we will always assume that the map $t\rightarrow \omega(t)$ is at least differentiable for $t\ge 0$.

In addition, we need the following definitions

\begin{definition}
Given a generic set of operators $\{a_{\alpha}\}_{\alpha}$, we say that the set is
\begin{enumerate}

\item \textit{self-adjoint}: if for any element $a_{\alpha}$ in the set,  there exists some index $\beta_{\alpha}$ such that $a_{\beta_{\alpha}}=a_{\alpha}^{\dagger}$;

\item \textit{irreducible}: if the commutant of $\{a_{\alpha}\}_{\alpha}$  consists only of operators proportional to the identity, namely
\begin{equation}
\{a_{\alpha}\}_{\alpha}^{'}=\{X\in\mathscr{B}| [X,a_{\alpha}]=0, \forall \alpha \}=\Span\{1_H\},
\end{equation}
Equivalently, we can say that the commutant is isomorphic to the space of complex numbers, i.e. $\{a_{\alpha}\}_{\alpha}^{'}\simeq \mathbb{C}$.
\end{enumerate} 
\end{definition}

In the time-independent case, the fundamental result upon which this work is based is (\cite{Spohn_2,SpohnReview, Frigerio_1})

\begin{theorem}[\textbf{Spohn-Frigerio}]
\label{th_spohnfrig}
Let $\mathcal{L}$ be the generator of a CPTP semigroup and let $\mathcal{J}=\{L_{\alpha}\}_{\alpha}$ be the set of jump operators. If 
\begin{equation}
\label{eq_condSF}
\begin{aligned}
\mathcal{J} & \text{  is self-adjoint} \\
\mathcal{J}^{'}&= \Span\{1_H\}    \\
\end{aligned}
\end{equation}
then there exists a faithful (i.e. full-rank) steady state $\omega$ and the semigroup is relaxing, namely  $\lim\limits_{t\rightarrow\infty}e^{t\mathcal{L}}(\rho_0)=\omega$, $\forall\rho_0 \in\mathscr{S}_1^+$.
\end{theorem}

An alternative proof of the theorem, based on the methodology developed in this paper, can be found in Appendix \ref{sec:app_Heisenberg}.

\subsection{\label{sec: main_results}Main results}

The main results of this work provide an extension of Theorem \ref{th_spohnfrig} to the time-dependent  case.

\begin{theorem}
\label{th: main_1}
Let $\{ \mathcal{L}_t\}_{t\ge 0}$ be a family of bounded GKLS generators as in Eq.\eqref{markov_eq}, with $t\rightarrow \mathcal{L}_t$ continuous map and let $\mathcal{J}=\{L_{\alpha}(t)\}_{\alpha}$ be the set of jump operators and $\{\gamma_{\alpha}(t)\}_{\alpha}$ the set of dissipation rates. If $\forall t \in [t_0,t_1]\subseteq \mathbb{R}^+$, the conditions in Eq.\eqref{eq_condSF} are satisfied, then $\ker(\mathcal{L}_t)=\Span\{\omega(t)\}$, with $\omega(t)>0$  and the restriction to the subspace of traceless operators $\mathscr{S}_0$ of the dynamical map, $\tilde{\Lambda}_{t,t_0}$, satisfies 
\begin{equation}
\label{eq_boundLambda}
\begin{aligned}
& \exp\Big(-\int\limits_{t_0}^t ds f(s)\Big)\sqrt{\frac{m(\omega(t))}{ N} }\\
&\le||\tilde{\Lambda}_{t,t_0}||_{1,1} \\
&\le \exp\Big(-\int\limits_{t_0}^t ds F(s)\Big) \sqrt{\frac{N}{m(\omega(t_0))} } ,
\end{aligned}
\end{equation}
where 
\begin{equation}
\label{eq_boundsF}
\begin{aligned}
F(t)=&\frac{1}{2}\Big[ \frac{m(\omega(t))}{M(\omega(t))} \lambda_\text{min}(t)  +\frac{m(\dot{\omega}(t))}{m(\omega(t))}  \Big],  \\
f(t)=&  \frac{1}{2}\Big[ \frac{M(\omega(t))}{m(\omega(t))} \lambda_{\text{max}}(t)  +\frac{M(\dot{\omega}(t))}{m(\omega(t))}   \Big], \\
\end{aligned}
\end{equation}
and 
\begin{equation}
\label{eq_lambdas}
\begin{aligned}
\lambda_{\text{max}}(t)&= \sup\limits_{a\perp 1_H}\sum_{\alpha}\gamma_{\alpha}(t)\frac{||[L_{\alpha}(t),a]||_2^2}{||a||_2^2}, \\
\lambda_{\text{min}}(t)&= \inf\limits_{a\perp 1_H}\sum_{\alpha}\gamma_{\alpha}(t)\frac{||[L_{\alpha}(t),a]||_2^2}{||a||_2^2}>0. \\
\end{aligned}
\end{equation}
\end{theorem}
We observe that every quantity in Eq.\eqref{eq_boundsF} is positive, but $m(\dot{\omega}(t))<0$ (see the discussion in the proof in Section \ref{sec:proof}).
The functions $f(t),F(t)$ in Eq.\eqref{eq_boundsF} provide the time-dependent rates relative to the exponential contraction of the dynamics on the traceless subspace and become constant in the time-independent case. These rate functions contain one contribution stemming from the variation of the instantaneous steady state, via $\dot{\omega}(t)$, and another one directly related to the dissipative part of the generator, via $\lambda_{\text{max}}(t), \lambda_{\text{min}}(t)$ introduced in Eq.\eqref{eq_lambdas}. In particular, the latter quantities are eigenvalues of a superoperator associated to the GKLS generator (as explained in detail in Section \ref{sec:proof}), whereas in Section \ref{sec: lambda} an explicit derivation based on the microscopic dynamics is provided).

\begin{theorem}
\label{th: main_2}
Let $\{ \mathcal{L}_t\}_{t\ge 0}$ be a family of bounded GKLS generators as in Eq.\eqref{markov_eq}. Let us assume that there exists $t_0\ge 0$ such that, $\forall t\ge t_0$, $\mathcal{L}_t$ satisfies the hypothesis of Theorem \ref{th: main_1} and
\begin{equation}
\label{eq_condrelax}
\lim\limits_{t\rightarrow \infty}\int\limits_{t_0}^t ds F(s)=\infty.
\end{equation}
Then the dynamics is weakly relaxing and the asymptotic trajectory is given by
\begin{equation}
\theta(t)=\frac{1_H}{N}+\frac{1}{N}\int\limits_0^t ds \Lambda_{t,s}\mathcal{L}_s(1_H).
\end{equation}

\end{theorem}

Physically, the condition in Eq.\eqref{eq_condrelax} simply states that, provided the system experiences "enough dissipation" during the evolution, then the relaxation is guaranteed. We remark that, for such condition to be satisfied, it is not necessary that the function $F(t)$ is strictly positive at any time, as long as the integral diverges. Moreover, since Eq.\eqref{eq_condrelax} is a sufficient condition for relation, we emphasize that weaker conditions may exist, especially for specific classes of dynamics.

Naturally, we underline that, in the case of time-independent generator, our results contain Theorem \ref{th_spohnfrig} as limiting case (see also Appendix \ref{sec:app_Heisenberg}). 

Furthermore, we have the following 

\begin{corollary}
\label{cor_charact}
If the hypotheses of Theorem \ref{th: main_2} are satisfied and in addition $\mathcal{L}_t(1_H)=0$, then the dynamics is strongly relaxing to the maximally mixed state.
\end{corollary}

\begin{corollary}
\label{cor_nonMarkov}
Let $\{\Lambda_{t,0}\}_{t,0}$ be an arbitrary dynamical map, described by a time-local master equation $\ddt{\Lambda_{t,0}}=\mathcal{K}_t \Lambda_{t,0}$. If there exists $t_0\ge 0$ such that, $\forall t\ge t_0$, the generator $\mathcal{K}_t$ is in the GKLS form and satisfies the hypotheses in Theorem \ref{th: main_2}, then the dynamics is weakly relaxing.
\end{corollary}

Since time-local master equations with a GKLS-like structure represent a general form of non-Markovian evolution (\cite{Hall}), the latter corollary simply requires the positivity of the dissipation rates asymptotically, that is, the dynamics becomes CP-divisible for sufficiently large times.

This completes the general theory. The interested reader can find the proof in Section \ref{sec:proof}, whereas we refer to Section \ref{sec:appl} for applications.

\section{\label{sec:proof}Proof of the main results}

This section is devoted to proving the main results of this paper.

We start by introducing a modified version of the Kadison-Schwarz inequality (\cite{Kadison1952}) for CPTP maps. 

\begin{lemma}
\label{lemma_tpKadison}
Let $\Lambda: \mathscr{B}\rightarrow \mathscr{B}$ be a CPTP map with a full-rank fixed point $\omega>0$, then $\forall a\in\mathscr{B}$,
\begin{equation}
\label{eq_tpKad}
\Lambda(a^{\dagger} \sqrt{\omega})\omega^{-1} \Lambda(\sqrt{\omega}a)\le \Lambda(a^{\dagger} a).
\end{equation}
\end{lemma}
\begin{proof}
Since the map $\Lambda$ is 2-positive, we have that 
\begin{equation}
(\Lambda\otimes 1_2)(A^{\dagger}A)\ge 0,
\end{equation}
for any $A\in \mathscr{B}\otimes \mathbb{C}^2$, thus if we consider 

\begin{equation}
A=\left[
\begin{array}{cc}
a & \sqrt{\omega} \\
0 & 0 \\
\end{array}
\right],
\end{equation}

we obtain
\begin{equation}
(\Lambda\otimes 1_2)(A^{\dagger}A)=\left[
\begin{array}{cc}
\Lambda(a^{\dagger} a) &  \Lambda(a^{\dagger}\sqrt{\omega})  \\
\Lambda(\sqrt{\omega}a) &  \Lambda(\omega)\\
\end{array}
\right],
\end{equation}
where $\Lambda(\omega)=\omega$. But, utilizing the general decomposition 
\begin{equation}
\begin{aligned}
\left[
\begin{array}{cc}
Y & X^{\dagger} \\
X & Z \\
\end{array}
\right]=&
\left[
\begin{array}{cc}
1 & X^{\dagger} Z^{-1}  \\
0 & 1 \\
\end{array}
\right]
\left[
\begin{array}{cc}
Y-X^{\dagger} Z^{-1} X & 0 \\
0 & Z\\
\end{array}
\right] \\
& \times
\left[
\begin{array}{cc}
1 & 0 \\
Z^{-1} X & 1 \\
\end{array}
\right],
\end{aligned}
\end{equation}
we clearly see that positivity implies 
\begin{equation}
Y-X^{\dagger} Z^{-1}X= \Lambda(a^{\dagger} a)-\Lambda(a^{\dagger} \sqrt{\omega})\omega^{-1} \Lambda(\sqrt{\omega}a )\ge 0.
\end{equation}

\end{proof}
If $\Lambda$ is a unital map, then for $\omega=1_H$ the standard Kadison-Schwarz inequality is recovered.

This result implies some immediate consequences, summarized by the following lemma

\begin{lemma}
\label{lemma_corollary}
Let $\{\Lambda_t=e^{t\mathcal{L}}\}_{t\ge 0}$ be  a CPTP semigroup with a faithful steady state $\omega$, then 
\begin{enumerate}
 \item The generator satisfies the inequality 
\begin{equation}
\begin{aligned}
\label{eq_kadL}
&a^{\dagger} \omega^{-1/2}\mathcal{L}(\omega^{1/2}a)+\mathcal{L}(a^{\dagger} \omega^{1/2})\omega^{-1/2}a -\mathcal{L}(a^{\dagger}a) \\
&=-\sum_{\alpha}\gamma_{\alpha}[L_{\alpha},a^{\dagger}\omega^{-1/2}]\omega[L_{\alpha},a^{\dagger}\omega^{-1/2}]^{\dagger}\le 0.
\end{aligned}
\end{equation} 
\item Let us consider the map
\begin{equation}
\label{eq_resc}
\Lambda_t^{\natural}=\Omega^{-1}\Lambda_t \Omega=e^{t \Omega^{-1}\mathcal{L}\Omega},
\end{equation}
where $\Omega(a)=\omega^{1/2} a$, then $\Lambda_t^{\natural}$ is a contraction in the Hilbert-Schmidt norm, namely $||\Lambda_t^{\natural}(a)||_2\le ||a||_2$.

\item Let $\mathscr{B}_{\omega}=\Span\{\omega^{1/2}\}^{\perp}$ be the orthogonal space to $\omega^{1/2}$, then this is an invariant subspace under $\Omega^{-1}\mathcal{L}\Omega(\mathscr{B}_{\omega})\subseteq \mathscr{B}_{\omega}$.
\end{enumerate}

\end{lemma}

\begin{proof}

\begin{enumerate}

\item Eq.\eqref{eq_kadL} follows from a direct calculation, making use of the fact that $\mathcal{L}(\omega)=0$, whereas the inequality in the operator sense follows directly from the expansion for $t\rightarrow 0$ of Eq.\eqref{eq_tpKad}.

\item  Let $a\in\mathscr{B}$, then Eq.\eqref{eq_tpKad} can be expressed as $(\Lambda_t^{\natural}(a))^{\dagger} \Lambda^{\natural}_t(a)\le \Lambda_t(a^{\dagger}a)$ and we have 
\begin{equation}
\begin{aligned}
&||\Lambda^{\natural}_t(a)||_2^2=\Tr((\Lambda_t^{\natural}(a))^{\dagger} \Lambda^{\natural}_t(a) ) \\
&\le \Tr(\Lambda_t(a^{\dagger}a))= \Tr(a^{\dagger}a)=||a||_2^2.
\end{aligned}
\end{equation}

\item 
The invariant subspace property follows from the fact that $\Tr(\omega^{1/2}\Omega^{-1}\mathcal{L}\Omega(a))=\Tr(\mathcal{L}\Omega(a))=0$,  $\forall a\in\mathscr{B}$, since $\mathcal{L}$ is a trace-annihilating operator.

\end{enumerate}

\end{proof}

Moreover, we will make use of this important result (\cite{Gronwall,Bainov} )

\begin{lemma}{\textbf{(Grönwall's inequality)}}
\label{lemma_gron}
Let $I=[a,b]\subset\mathbb{R}$ be an interval, with $a<b$ and $b$ may also be infinite. Let $\beta, y: I\rightarrow \mathbb{R}$ continuous functions, with $y$ differentiable on $I_0$, the interior part of $I$, and such that $\frac{dy}{dt}(t)\le \beta(t) y(t),\forall t\in I_0$, then 
$y(t)\le y(a)\exp\int\limits_a^t ds \beta(s)$, $\forall t\in I$.
\end{lemma}

We are now ready to prove Theorem \ref{th: main_1}.

\begin{proof}

Let $\mathcal{L}_t$ be a GKLS generator as in the hypothesis, then for any fixed $t$, Theorem \ref{th_spohnfrig} guarantees the existence of a unique full-rank element in the kernel $\omega(t)>0$ (see Appendix \ref{sec:app_Heisenberg} for an alternative proof). For simplicity, we can assume  $t_0=0$, without loss of generality.

Let $\tilde{\Lambda}_{t,s}: \mathscr{S}_0\rightarrow \mathscr{S}_0$ be the restriction of the dynamical map to the subspace of traceless (self-adjoint) operators, then $\tilde{\rho}(t)=\tilde{\Lambda}_{t,0}(\tilde{\rho}_0)$.
We consider the change of variable given by $\tilde{\rho}(t)=\omega^{1/2}(t)\sigma(t)\equiv \Omega_t(\sigma(t))$, which leads to the differential problem
\begin{equation}
\label{eq_sigma}
\ddt{\sigma(t)}=\Omega_t^{-1}\mathcal{L}_t \Omega_t (\sigma(t)) + \ddt{\Omega_t^{-1}}\Omega_t (\sigma(t)).
\end{equation}
By construction, $\sigma(t)\in\mathscr{B}_{\omega(t)}=\Span\{\omega^{1/2}(t)\}^{\perp}$ and we observe that in general it is not self-adjoint. From Eq.\eqref{eq_sigma}, we can derive
\begin{equation}
\begin{aligned}
\ddt{||\sigma(t)||_2^2}=& 2\Re\Big(\sigma(t),\Omega^{-1}_t\mathcal{L}_t\Omega_t(\sigma(t))\Big)_{HS} \\
&- \Big(\omega^{-1/2}(t)\sigma(t),\frac{d \omega(t) }{dt}\omega^{-1/2}(t)\sigma(t) \Big)_{HS}, \\
\end{aligned}
\end{equation}
where we employed the identity
\begin{equation}
\begin{aligned}
&\omega^{-1/2}(t)\frac{d\omega(t)}{dt}\omega^{-1/2}(t)\\ &=\omega^{-1/2}(t)\Big\{ \frac{d\omega^{1/2}(t)}{dt}, \omega^{1/2}(t) \Big\} \omega^{-1/2}(t) \\
&= - \omega^{1/2}(t)\frac{d\omega^{-1/2}(t)}{dt}-\frac{d\omega^{-1/2}(t)}{dt}\omega^{1/2}(t).
\end{aligned}
\end{equation}
Making use of the expression obtained by taking the trace in Eq.\eqref{eq_kadL}, together with the fact that $\Tr\mathcal{L}(\sigma^{\dagger}\sigma)=0$, we arrive at
\begin{equation}
\label{eq_sigmanorm}
\begin{aligned}
\ddt{||\sigma(t)||_2^2}=&-\sum_{\alpha}\gamma_{\alpha}(t)|| [L_{\alpha}(t),\sigma^{\dagger}(t)\omega^{-1/2}(t)]\omega^{1/2}(t)||_{2}^2\\
&- \Big(\omega^{-1/2}(t)\sigma(t),\frac{d \omega(t) }{dt}\omega^{-1/2}(t)\sigma(t) \Big)_{HS} \\
\end{aligned}
\end{equation}
We want to provide an upper and lower bound to the right-hand side. To this end, it is convenient to utilize the relation (we recall the notation $m(X)=\inf \Sp(X)$, $M(X)=\sup \Sp(X)$)
\begin{equation}
\label{eq_boundX}
m(X)||a||_2^2 \le \Tr(a^{\dagger} X a) \le M(X)||a||_2^2, \\
\end{equation}
for $X$ self-adjoint, as well as the identity
\begin{equation}
||[L_{\alpha},a]||_2^2= \Tr(a^{\dagger} [L_{\alpha}^{\dagger},[L_{\alpha},a]]),
\end{equation}
to obtain the inequalities
\begin{equation}
\label{eq_firstK}
\begin{aligned}
&-M(\omega(t))\Big(\sigma(t)^{\dagger}\omega^{-1/2}(t),\mathcal{K}_t(\sigma(t)^{\dagger}\omega^{-1/2}(t))\Big)_{HS} \\
& \le  -\sum_{\alpha}\gamma_{\alpha}(t)|| [L_{\alpha}(t),\sigma(t)^{\dagger}\omega^{-1/2}(t)] \omega^{1/2}(t)||_{2}^2 \\
&\le -m(\omega(t))\Big(\sigma(t)^{\dagger}\omega^{-1/2}(t),\mathcal{K}_t(\sigma(t)^{\dagger}\omega^{-1/2}(t))\Big)_{HS},
\end{aligned}
\end{equation}
where 
\begin{equation}
\label{eq_Kdef}
\mathcal{K}_t(a)=\sum_{\alpha}\gamma_{\alpha}(t)[L_{\alpha}^{\dagger}(t),[L_{\alpha}(t),a]].
\end{equation} 
$\mathcal{K}_t$ is a self-adjoint operator, hence its spectrum can be bounded by the supremum and infimum of the expression $\frac{(a,\mathcal{K}_t(a))_{HS}}{||a||_2^2}$. In particular, by second hypothesis in Eq.\eqref{eq_condSF}, $(a,\mathcal{K}_t(a))_{HS}=\sum_{\alpha}\gamma_{\alpha}(t)||[L_{\alpha}(t),a]||_2^2=0$ implies that $a\in \Span\{1_H\}$. Therefore, the infimum on the subset of operators $\sigma^{\dagger}(t)\omega^{-1/2}(t)$ must be strictly positive, given that $\sigma^{\dagger}(t)\omega^{-1/2}(t)\in \Span\{1_H\}$ implies that $\sigma^{\dagger}(t)\in \Span\{\omega^{1/2}(t)\}$, which is not possible by construction.
Hence, the expressions in Eq.\eqref{eq_lambdas} are obtained. 

Therefore, we have
\begin{equation}
\begin{aligned}
&-\frac{M(\omega(t)) \lambda_{\text{max}}(t) }{m(\omega(t))}  ||\sigma(t)||_2^2 \\
& \le  -\sum_{\alpha}\gamma_{\alpha}(t)|| [L_{\alpha}(t),\sigma(t)^{\dagger}\omega^{-1/2}(t)] \omega^{1/2}(t)||_{2}^2  \\
&\le -\frac{m(\omega(t)) \lambda_{\text{min}}(t) }{M(\omega(t))}||\sigma(t)||_2^2.
\end{aligned}
\end{equation}
For the other term in Eq.\eqref{eq_sigmanorm}, we can use the spectrum of $\dot{\omega}(t)$. Since $\Tr(\dot{\omega}(t))=0$, we have that $\inf \Sp(\dot{\omega}(t))\le 0$, $\sup \Sp(\dot{\omega}(t))\ge 0$, where the equality is clearly attained if and only if $\dot{\omega}(t)=0$. As a consequence, we have
\begin{equation}
\begin{aligned}
&- \frac{|m(\dot{\omega})(t)|}{m(\omega)(t)}  ||\sigma(t)||_2^2 \\
&\le (\omega^{-1/2}(t)\sigma(t),\frac{d \omega(t) }{dt}\omega^{-1/2}(t)\sigma(t)))_{HS} \\
&\le \frac{M(\dot{\omega})(t)}{m(\omega) (t)}  ||\sigma(t)||_2^2.
\end{aligned}
\end{equation}
Including these inequalities in Eq.\eqref{eq_sigmanorm} and utilizing Grönwall's Lemma \ref{lemma_gron}, we obtain
\begin{equation}
\label{eq_boundsigma}
\exp\Big(-\int\limits_0^t ds f(s) \Big) \le \frac{||\sigma(t)||_2}{||\sigma(0)||_2}\le  \exp\Big(-\int\limits_0^t ds F(s) \Big),
\end{equation}
where $f(t),F(t)$ are given in Eq.\eqref{eq_boundsF}.
Finally, due to the general property of the Schatten norms $||ab||_{1}\le ||a||_p ||b||_q$ with $1/p+1/q=1$, and utilizing $||a||_2\le ||a||_1\le \sqrt{\rk(a)}||a||_2\le \sqrt{N}||a||_2$, we have 
\begin{equation}
\begin{aligned}
||\tilde{\rho}(t)||_1 & \le \underbrace{||\omega^{1/2}(t)||_{1}}_{\le \sqrt{N}}\underbrace{||\sigma(t)||_{\infty} }_{\le ||\sigma(t)||_{2}}\le \\
&\le \exp\Big(-\int\limits_0^t F(s)ds \Big)  \sqrt{N} \underbrace{||\sigma(0)||_2}_{\le ||\sigma(0)||_1} \\
& \le  \exp\Big(-\int\limits_0^t F(s)ds \Big)\sqrt{N}  ||\omega^{-1/2}(0)||_{\infty} ||\tilde{\rho}(0)||_{1},
\end{aligned}
\end{equation}
and, similarly, for the lower bound we have the chain of inequalities
\begin{equation}
\begin{aligned}
||\tilde{\rho}(t)||_1\ge & ||\omega^{1/2}(t)\sigma(t)||_2 \\
\ge & \sqrt{m(\omega(t))}\exp\Big(-\int\limits_0^t ds f(s)\Big)||\omega^{-1/2}(0)\tilde{\rho}(0)||_2 \\
\ge & \sqrt{m(\omega(t))}\exp\Big(-\int\limits_0^t ds f(s)\Big)\frac{||\tilde{\rho}(0)||_2}{\sqrt{M(\omega(0))}} \\
\ge & \sqrt{m(\omega(t))}\exp\Big(-\int\limits_0^t ds f(s)\Big)\frac{||\tilde{\rho}(0)||_1 N^{-1/2}}{\sqrt{M(\omega(0))}},
\end{aligned}
\end{equation}
which leads to Eq.\eqref{eq_boundLambda}.

\end{proof}

Theorem \ref{th: main_2} follows as a consequence.

\begin{proof}
Let $\mathcal{L}_t$ be a time-dependent GKLS generator, we consider the differential problem given by
\begin{equation}
\ddt{\rho(t)}=\mathcal{L}_t(\rho(t)) 
\end{equation}
for generic $\rho_0=\rho(0)\in \mathscr{S}^+_1$. Since $\Tr(\rho(t))=1, \forall t\ge 0$,  we can uniquely decompose the density matrix, utilizing the Hilbert-Schmidt scalar product, as $\rho(t)=\frac{1_H}{N}+\tilde{\rho}(t)$, where $\tilde{\rho}(t)$ belongs to the subspace of traceless operators, which is orthogonal to $\Span\{1_H\}$. Plugging this  decomposition in the master equation for $\rho(t)$, it is not difficult to see that the general solution can be expressed in the form
\begin{equation}
\label{eq_rhodecomp}
\rho(t)=\Lambda_{t,0}(\tilde{\rho}_0)+\underbrace{\frac{1_H}{N}+\frac{1}{N}\int\limits_0^t ds \Lambda_{t,s}\mathcal{L}_s(1_H)}_{=\theta(t)}.
\end{equation}
The contribution $\theta(t)$ in Eq.\eqref{eq_rhodecomp} represents the asymptotic trajectory, which satisfies the time-dependent master equation $\ddt{\theta(t)}=\mathcal{L}_t(\theta(t))$ for $\theta(0)=1_H/N\in\mathscr{S}_1^+$, hence $\theta(t)\in\mathscr{S}_1^+$, $\forall t\ge 0$.
The first term contains the dependence on the initial state, via its projection onto the subspace of traceless operators $\mathscr{S}_0$, which is an invariant subspace under the TP map $\Lambda_{t,0}$. Therefore, let $\tilde{\Lambda}_{t,0} $ be the restriction of the dynamical map to the subspace $\mathscr{S}_0$, then in order to prove relaxation we have to show that 
\begin{equation}
\lim\limits_{t\rightarrow \infty}||\rho(t)-\theta(t)||=\lim\limits_{t\rightarrow \infty}||\tilde{\Lambda}_{t,0} (\tilde{\rho}_0)||=0,
\end{equation}
$\forall \tilde{\rho}_0$. 

Because the hypotheses in Theorem \ref{th: main_1} are satisfied for $t\ge t_0\ge 0$, we have that 
\begin{equation}
\begin{aligned}
||\tilde{\Lambda}_{t,0}||_{1,1} &\le ||\tilde{\Lambda}_{t,t_0}||_{1,1} \underbrace{||\tilde{\Lambda}_{t_0,0}||_{1,1}}_{\le 1} \\
 &\le\exp\Big( -\int\limits_{t_0}^t ds F(s)\Big)  \sqrt{\frac{N}{m(\omega(t_0))} }\rightarrow 0,
\end{aligned}
\end{equation}
which provides the relaxation.

\end{proof}

\section{\label{sec:appl}Application to driven N-level systems}

In order to illustrate some instances of the theory so far exposed, we will discuss how the bounds derived in Section \ref{sec:summary} apply to specific master equations. For concreteness, we will consider the microscopic dynamics described  in \cite{DiMeglio2024}, although other derivations can also be utilized (as for instance \cite{Albash2012, Dann2018, Mozgunov2020}). To begin with, we will keep our derivation completely general and obtain some  properties valid for fairly arbitrary systems.

Let us consider a $N$-level system subject to driving and interacting with an environment, as described by the microscopic Hamiltonian
\begin{equation}
H(t)=H_S(t)+H_E+H_I,
\end{equation}
where $H_S(t)=\sum_n E_n(t)|n(t)\rangle\langle n(t)|$ is an arbitrary time-dependent Hamiltonian for the system, with $t\rightarrow H_S(t)$ analytic function. Moreover, we assume that the spectrum contains no degeneracies, at least for sufficiently large $t$. Clearly, since we are interested in the long-time dynamics, degeneracies and level crossings at the early stages of the evolution are not important for our analysis. 

As customary, the environment is given by a continuum of bosonic modes with annihilation and creation operators $b_j(k),b_j^{\dagger}(k)$, where $[b_i(k),b_j^{\dagger}(p)]=\delta_{ij}\delta(k-p)$, hence the free environment Hamiltonian and the interaction term read
\begin{equation}
\label{eq_bosonicbath}
\begin{aligned}
H_E=&\sum_j \int\limits_0^{\infty} dk \omega(k)b_j^{\dagger}(k)b_j(k), \\
H_I=&\sum_j A_j\otimes \underbrace{\int\limits_0^{\infty} dk g_j(k)(b_j(k)+b_j^{\dagger}(k))}_{=B_j}
\end{aligned}
\end{equation}
with $A_j=A^{\dagger}_j$. The index $j$ is typically associated to some degree of freedom in the system, as for instance the $j$-th site of a lattice. More general forms of interaction are equally admissible.

Following the prescriptions in \cite{DiMeglio2024}, one derives a time-dependent master equation in the interaction picture, with generator
\begin{equation}
\label{eq_tdme}
\begin{aligned}
\mathcal{L}_{I,t}(\rho)=&-\imath[H_{I,LS}(t),\rho] \\
&+\sum_{\alpha}\gamma_{\alpha}(t) \Big(L_{I,\alpha}(t)\rho L_{I,\alpha}^{\dagger}(t) \\
&-\frac{1}{2}\{L_{I,\alpha}^{\dagger}(t) L_{I,\alpha}(t),\rho\} \Big),
\end{aligned}
\end{equation}
where $\alpha=(j,n, m)$ is a composite index. The dissipation rates are obtained from the identity
\begin{equation}
\label{eq_dissratesdef}
\begin{aligned}
\gamma_{(j,n,m)}(t)= \int\limits_{-\infty}^{\infty} dx  \Tr_E\Big(e^{\imath H_E x}B_j e^{-\imath H_E x} B_j \rho_{E}\Big) e^{\imath x(\Omega_{nm}(t))} \\
\end{aligned}
\end{equation}
with $\rho_E$ the initial state of the environment, while
\begin{equation}
\label{eq_lambshift}
H_{I,LS}(t)=\sum_{j,n,m} S_{(j,n,m)}(t) L_{I,(j,n,m)}^{\dagger}(t)L_{I,(j,n,m)}(t)
\end{equation}
and $S_{(j,n,m)}(t) $ are extracted from the one-sided Fourier transform of the environment correlation functions $\Tr_E\Big(e^{\imath H_E x}B_j e^{-\imath H_E x} B_j \rho_{E}\Big)$ (see \cite{DiMeglio2024} for more details).
The jump operators stem from the decomposition
\begin{equation}
\label{eq_jumpandhls}
\begin{aligned}
V^{\dagger}(t) A_j  V(t)=\sum_{n,m} \exp\Big(-\imath \int\limits_0^t ds (\underbrace{E_n(s)-E_m(s)}_{=\Omega_{nm}(s)})\Big) \\
\times \underbrace{|m(0)\rangle\langle m(t)|A_j|n(t)\rangle\langle n(0)|}_{=L_{I,(j,n,m)}(t)},
\end{aligned}
\end{equation}
with $V(t)=\sum_n \exp\Big(-\imath\int\limits_0^t ds E_n(s)\Big)|n(t)\rangle\langle n(0)|$. 

We underline that such decomposition must be performed in a way that the contributions $|m(0)\rangle\langle m(t)|A_j|n(t)\rangle\langle n(0)|$ characterized by the same Bohr frequency $\Omega_{nm}(t)$ are grouped together, in order to ensure the correct canonical GKLS form (for the time-independent case, this fact is  discussed in standard textbooks such as \cite{Breuer}). In particular, this implies that no degeneracies occur in the set $\{\Omega_{nm}(t)\}_{n,m}$. 

Therefore, we can identify three types of jumps depending on the Bohr frequencies, namely:
\begin{enumerate}
\item $\Omega_{nm}=0$. In this case, the absence of degeneracies guarantees that the jumps are diagonal in the eigenbasis $\{|n\rangle\equiv |n(0)\rangle\}_n$. Introducing the index $\alpha_0=(j,0,0)$, they read
\begin{equation}
\label{eq_jumpdiag}
L_{I,\alpha_0}=\left(
\begin{array}{cccc}
* & 0 &  0 & 0\\
0 & * & 0 & 0\\
0 & 0 &  \ddots  & 0 \\
0 & 0 &  0 &  *\\
\end{array}
\right).
\end{equation}
These jumps are physically responsible for dephasing in the eigenbasis of $H_S(0)$.
\item $\Omega_{nm}>0$. These jumps labeled by $\alpha_+=(j,n,m)$ are upper-triangular with zeros on the diagonal, namely 
\begin{equation}
\label{eq_jumpoffdiag}
L_{I,\alpha_+}=\left(
\begin{array}{cccc}
0 & * & *  & * \\
0 & 0 & *  &  * \\
0 & 0 &  \ddots & *  \\
0 & 0 &  0 & 0\\
\end{array}
\right).
\end{equation}

We observe that the hypothesis of non-degeneracy requires that the same jump cannot 
induce transitions from a given state $|n\rangle$ to different $|i\rangle,|j\rangle$, for $i\neq j$, given that the condition $\Omega_{ni}=\Omega_{nj}$ implies  $E_i=E_j$. Therefore, the upper-triangular structure 
is further simplified by the fact that each row and column can contain at most one non-zero entry.
Naturally, even in the presence of degeneracies, this structure can be achieved if the corresponding matrix elements of $A_j$ are zero. In the following, we will additionally assume that each jump $L_{I,\alpha_+}$ has exactly one non-zero entry, namely it can induce incoherent transitions between two specific states $|n_{\alpha}\rangle, |m_{\alpha}\rangle $, for $n_{\alpha}\neq m_{\alpha}$. 
\item $\Omega_{nm}<0$. These are given by $\Omega_{nm}=-\Omega_{mn}$, and  $L_{I,\alpha_-}=L_{I,\alpha_+}^{\dagger}$, i.e. lower-triangular with zeros on the main diagonal, where $\alpha_-=(j,n,m)$ corresponds to $\alpha_+=(j,m,n)$. 
\end{enumerate}

We observe that the analysis of the asymptotic dynamics in the interaction picture provides automatically information on the behavior of the state in the Schrödinger picture. In particular, if  the system density matrix in interaction picture $\rho_{I, S}(t)\rightarrow \theta_{I}(t)$ asymptotically, then $\rho_S(t)\rightarrow U_S(t)\theta_{I}(t)U_S^{\dagger}(t)$, where $U_S(t)$ is the time-evolution operator generated by $H_S(t)$. Therefore,  we can simply address the problem for $\rho_{I, S}(t)$, taking advantage of the simpler form of the interaction picture generator.
Furthermore, we notice that similar results as the ones discussed in this section hold also for the adiabatic master equation \cite{Albash2012}.

In order to address the asymptotic relaxation by means of our results, we need to characterize three main features: the instantaneous steady state $\omega_{I}(t)$, the conditions in Eq.\eqref{eq_condSF} on the jump operators and the quantities in Eq.\eqref{eq_lambdas}. We will tackle these points in each of the following subsections.

\subsection{\label{sec: omega}Instantaneous steady state}

The results we are going to prove are related to the structure of the instantaneous steady state $\omega_{I}(t)$ and its spectrum.
\begin{lemma}
\label{lemma_omega_1}
Let $\mathcal{L}_{I,t}$ be a GKLS generator as in Eq.\eqref{eq_tdme} and let assume that the conditions in Eq.\eqref{eq_condSF} are satisfied. Then the unique instantaneous steady state is
\begin{equation}
\label{eq_omegastruct}
\omega_{I}(t)=\sum_n \omega_n(t) P_n,
\end{equation}
where $P_n=|n(0)\rangle\langle n(0)|$ are the orthogonal projectors given by the eigendecomposition of $H_S(0)$ and the eigenvalues are
 the solution of the homogeneous equation
\begin{equation}
\label{eq_homomega}
\sum_m \omega_m(t)\Big(\Tr(\Phi_t^{\ddagger}(P_n)P_m)-\delta_{nm}\Tr(\Phi_t(P_n))\Big)=0,
\end{equation}
with $\Phi_t(a)=\sum_{\alpha}\gamma_{\alpha}(t) L_{I,\alpha}(t)a L_{I,\alpha}^{\dagger}(t)$.
In particular
\begin{equation}
\label{eq_boundomegasp}
\frac{m(\Phi_t^{\ddagger}(P_n))}{\Tr(\Phi_t(P_n))} \le \omega_n(t)\le \frac{M(\Phi_t^{\ddagger}(P_n))}{\Tr(\Phi_t(P_n))}.
\end{equation}

\end{lemma}

\begin{proof}
Let us prove Eq.\eqref{eq_omegastruct} first.
One can easily see that the generator  in Eq.\eqref{eq_tdme} is invariant under any unitary superoperator in the form $\mathcal{W}_t(a)=W(t)aW^{\dagger}(t)$, where $W(t)=\sum_n e^{-\imath \alpha_n(t)}P_n$ and $\alpha_n(t)\in\mathbb{R}$, given that 
\begin{equation}
\mathcal{W}_t(L_{I,(j,n,m)}(t))=e^{\imath(\alpha_n(t)-\alpha_m(t))} L_{I,(j,n,m)}(t).
\end{equation}
But, since $[\mathcal{L}_{I,t},\mathcal{W}_t]=0$, this implies that $\mathcal{W}_t(\omega_{I}(t))\in \ker\mathcal{L}_{I,t}$, namely $\mathcal{W}_t(\omega_{I}(t))=\omega_{I}(t)$ (due to trace preservation).
Hence,  $[W(t),\omega_{I}(t)]=0$ and thus $\omega_{I}(t)=\sum_n \omega_n(t)P_n$.

Now we move to Eq.\eqref{eq_homomega}. Because $[H_{I,LS}(t),P_n]=0$ by construction (see Eq.\eqref{eq_lambshift}), then 
\begin{equation}
\mathcal{L}_{I,t}(\omega_{I}(t))=\Phi_t(\omega_{I}(t))-\frac{1}{2}\{ \Phi_t^{\ddagger}(1_H),\omega_{I}(t)\}.
\end{equation}
Due to the conditions in Eq.\eqref{eq_condSF}, we can prove that $\Phi_t^{\ddagger}(1_H)>0$. In fact, if there exists $\psi\in\mathscr{H}$ such that $\langle\psi|\Phi_t^{\dagger}(1_H)|\psi\rangle=0$, then $L_{I,\alpha}(t)|\psi\rangle=0, L_{I,\alpha}^{\dagger}(t)|\psi\rangle=0$, $\forall\alpha$, which implies that the commutant contains the non-trivial element $|\psi\rangle\langle\psi|$, that is a contradiction. Therefore, we have that
\begin{equation}
\begin{aligned}
0 &=\Tr\Big(P_n\Big(\Phi_t(\omega_{I}(t))-\frac{1}{2}\{ \Phi_t^{\ddagger}(1_H),\omega_{I}(t)\}\Big) \Big) \\
& =\Tr(\Phi_t^{\ddagger}(P_n)\omega_{I}(t))-\omega_n(t)\Tr(P_n\Phi_t^{\ddagger}(1_H)),
\end{aligned}
\end{equation}
which implies Eq.\eqref{eq_homomega} and Eq.\eqref{eq_boundomegasp}.
\end{proof}

In particular, this lemma shows that the Hamiltonian term $H_{I,LS}(t)$ does not contribute to determining the instantaneous steady state.

Furthermore, when $\rho_E$ is a thermal state and the environment correlation functions satisfy the KMS condition, the generator $\mathcal{L}_{I,t}$ always possesses a relevant type of zero eigenoperator, which is a thermal state.

\begin{lemma}
\label{lemma_fixedth}
Let $\mathcal{L}_{I,t}$ be the GKLS generator in Eq.\eqref{eq_tdme} and let us assume that the dissipation rates in Eq.\eqref{eq_dissratesdef} satisfy the detailed balance condition
\begin{equation}
\label{eq_detailedbalance}
\gamma_{(j,n,m)}(t)=e^{\beta (E_n(t)-E_m(t))}\gamma_{(j,m,n)}(t).
\end{equation}
Then, the Gibbs state 
\begin{equation}
\omega_{I,\text{th}}(t)=\frac{e^{-\beta H_{\text{avg}}(t)}}{\Tr_Se^{-\beta H_{\text{avg}}(t)}},
\end{equation} 
with respect to the Hamiltonian
\begin{equation}
\label{eq_avgH}
H_{\text{avg}}(t)=\sum_n E_n(t)|n(0)\rangle\langle n(0)|,
\end{equation}
is an instantaneous steady state.
\end{lemma}

\begin{proof}
The proof follows the standard argument for the time-independent case (see \cite{Rivas, Breuer}). One can easily see that the jump operators $L_{I,(j,n,m)}$ in Eq.\eqref{eq_jumpandhls} satisfy the relation
\begin{equation}
L_{I,(j,n,m)}\omega_{I,\text{th}}=e^{-\beta(E_n(t)-E_m(t))}\omega_{I,\text{th}}L_{I,(j,n,m)},
\end{equation}
which, combined with Eq.\eqref{eq_detailedbalance}, provides that $\mathcal{L}_{I,t}(\omega_{I,\text{th}}(t))=0$.
\end{proof}

When the conditions in Eq.\eqref{eq_condSF} are satisfied, this result provides complete knowledge of the instantaneous steady state, in the ubiquitous scenario of a thermal state for the environment.

\subsection{\label{sec: jumps}Graph-theoretic formulation of the jump operators}

In this section,  we will provide a characterization of the conditions in Eq.\eqref{eq_condSF}. In order to ease the notation, in this and the next section the dependence on $t$ will be omitted and the jumps will be simply denoted with $L_{\alpha}$.

Assuming the specific structure discussed in Eqs.\eqref{eq_jumpdiag} and \eqref{eq_jumpoffdiag}, we can generally split the set of jumps into diagonal and off-diagonal operators, as $\mathcal{J}=\mathcal{J}_{\text{diag}}\cup\mathcal{J}_{\text{off}}$, where
\begin{equation}
\label{jump_sets}
\begin{aligned}
\mathcal{J}_{\text{diag}} &\subseteq \Span \{ |n\rangle\langle n|\}_{n\in I}, \\
\mathcal{J}_{\text{off}}&\subseteq \Span  \{ |n\rangle\langle m|\}_{(n,m)\in E}, \\
\end{aligned}
\end{equation}
with $I\subseteq \{1,...,N\}$, $E=\{(n,m)\in I\times I| n\neq m\}$. In particular, as already explained in the text below Eq.\eqref{eq_jumpoffdiag}, the off-diagonal jumps are assumed to have the specific structure
\begin{equation}
\label{jumpoff_struct}
L_{\alpha}=\xi_{\alpha}|n_{\alpha}\rangle\langle m_{\alpha}|,
\end{equation}
for $\xi_{\alpha}\in \mathbb{C}$. Despite its simple form, this choice is frequently recurrent in several models of dissipation.

\begin{lemma}
\label{lemma_jumps}
Let $\mathcal{J}=\mathcal{J}_{\text{diag}} \cup \mathcal{J}_{\text{off}}$ be a set of jump operators as in Eq.\eqref{jump_sets} and let us assume that $\mathcal{J}$ is self-adjoint, then $\mathcal{J}^{'}=\Span\{1_H\}$ if and only if $\mathcal{J}_{\text{off}}^{'}=\Span\{1_H\}$.
\end{lemma}

\begin{proof}
If $\mathcal{J}_{\text{off}}^{'}=\Span\{1_H\}$ then clearly $\mathcal{J}$ has trivial commutant too, therefore let us focus on the other implication. Since the commutant is invariant under linear combinations, we can equivalently address the problem by studying $\{ |n\rangle\langle n|\}_{n\in I}^{'}$ and $ \{ |n\rangle\langle m|\}_{(n,m)\in E}^{'}$. 
Let us assume that  $\mathcal{J}^{'}$ is trivial, but $\mathcal{J}_{\text{off}}^{'}$ is not, hence there exists some $a\notin \Span\{1_H\} $ that commutes with all the off-diagonal elements $|n\rangle\langle m|$ for $n\neq m$. Since $\mathcal{J}_{\text{off}}$ is also self-adjoint by hypothesis,  then both $|n\rangle\langle m|$ and $|m\rangle\langle n|$ are in the set.
Let us express $a=\sum_{pq}a_{pq}|p\rangle\langle q|$, the condition $[|n\rangle\langle m|,a]=[|m\rangle\langle n|,a]=0$ implies
\begin{equation}
\begin{aligned}
\delta_{np}a_{mq}=& a_{np}\delta_{mq}, \\
\delta_{nq}a_{pm}=& a_{nq}\delta_{pm}, \\
\end{aligned}
\end{equation}
which gives $a_{nn}=a_{mm}$ and that the off-diagonal entries in the rows and columns $n,m$ are zero, namely
\begin{equation}
\label{eq_amatrix}
a=\left(
\begin{array}{ccccccc}
 &  &  0 & & 0 & & \\
  &  & \vdots & & & & \\
  0 & \dots & a_{nn} &   & \dots & & 0\\
      & &  \vdots & & \vdots & &\\
 0   & &  \dots &  & a_{mm} & \dots  & 0\\
          & &  & & \vdots & &\\
    & &  0 & & 0 & & \\
\end{array}
\right).
\end{equation}
It is not difficult to check that such $a$ also commutes with any linear combination of  $ |n\rangle\langle n|$ and $|m\rangle\langle m|$. Therefore, extending the argument to all jumps in $\mathcal{J}_{\text{off}}$, one has that 
$a\in \mathcal{J}^{'}_{\text{diag}}$, which is a contradiction.
\end{proof}

This lemma allows us to focus our analysis on the off-diagonal jump operators, which provide a simple characterization of the conditions in Eq.\eqref{eq_condSF}.
In order to do so, it is necessary first to introduce a graph-theoretic structure associated to the set of jumps.
Any reader who is not familiar with standard nomenclature and notions in graph theory can find in Appendix \ref{app: graph} a self-contained introduction to the topic and the results used in this work.
We emphasize that the graph-theoretic formulation introduced in this section serves as a tool to analyze the irreducibility condition for the specific class of microscopically derived generators considered here. The general results of Theorems \ref{th: main_1} and \ref{th: main_2} do not rely on this construction.

\begin{definition}
\label{def: jumpsgraph}
Let 
\begin{equation}
\Big\{ O_{\alpha}=\sum_{(n,m)\in E_{\alpha}} (O_{\alpha})_{nm}|n\rangle\langle m|\Big\}_{\alpha}
\end{equation}
be  a set of operators, where $E_{\alpha}\subseteq \{1,...,N\}\times \{1,...,N\} $ and $(O_{\alpha})_{nm}\neq 0$. The weighted graph $G=(V,E,W)$ associated to the set is given by
\begin{enumerate}
    \item $V=\{ |n\rangle \}_{n=1,...,N}$ set of vertices;
    \item $E=\bigcup\limits_{\alpha} E_{\alpha}$ set of edges;
    \item $W: (n,m)\in E\rightarrow \sum_{\alpha}| (O_{\alpha})_{nm}|^2\in \mathbb{R}^+$ weight function.
\end{enumerate}
\end{definition}
In the following, the set of operators under analysis will be given by the jump operators in Eq.\eqref{jump_sets}. In particular, we notice that the diagonal jumps correspond to loops in graph, whereas the off-diagonal subset is responsible for connections between different nodes.  Because of Lemma \ref{lemma_jumps},  we can simply neglect such loops and restrict our analysis to simple graphs. Moreover, if both pairs $(n,m)$ and $(m,n)$ are in $E$, i.e. $\mathcal{J}$ is self-adjoint, then  the adjacency relation of $G$ is symmetric, i.e. the graph is undirected.

We can now prove the following result
\begin{lemma}
\label{lemma_irrgraph}
Let  $\mathcal{J}_{\text{off}}$ be the set of off-diagonal jump operators in Eq.\eqref{jump_sets} and let $G$ be its associated graph. If $\mathcal{J}_{\text{off}}$ is self-adjoint and irreducible, then $G$ is undirected and connected.

\end{lemma}

\begin{proof}

The relation between self-adjointness of $\mathcal{J}_{\text{off}}$ and  the property that $G$ is undirected was already discussed. Now let us assume that $\mathcal{J}_{\text{off}}$ is irreducible but $G$ is not connected. Upon reordering of the set of vertices (i.e., of the basis $|n\rangle$), the graph can be expressed in terms of two disconnected graphs $G_1,G_2$ with  nodes $V_1,V_2$ such that $V_1\cap V_2=\emptyset, V_1\cup V_2=V$. This implies that the adjacency matrix $A_G$ takes the block-diagonal form
\begin{equation}
\label{eq_Agblock}
A_G=\left[
\begin{array}{cc}
(A_G)_1 & 0 \\
0 & (A_G)_2 \\
\end{array}
\right],
\end{equation}
where the entries in the off-diagonal blocks are $(A_G)_{nm}=0$, for $n\in V_1, m\in V_2$ and vice versa. But, since $(A_G)_{nm}=\sum_{\alpha}|(L_{\alpha})_{nm}|^2=0$ implies that  $(L_{\alpha})_{nm}=0$, $\forall\alpha$, thus each jump operator has also the block-diagonal form
\begin{equation}
L_{\alpha}=\left[
\begin{array}{cc}
(L_{\alpha})_1 & 0 \\
0 & (L_{\alpha})_2 \\
\end{array}
\right],
\end{equation}
which implies that any operator $a=a_1 1_1 \oplus a_2 1_2$, with $a_1,a_2\in \mathbb{C}$ commutes with all jumps and thus $\mathcal{J}_{\text{off}}$ is not irreducible.

\end{proof}

Lemmas \ref{lemma_jumps} and \ref{lemma_irrgraph} provide a simple criterion to address the conditions in Eq.\eqref{eq_condSF}, making use of a graph analysis. We observe that, a relation between connectivity of graphs associated to the GKLS generator and ergodic properties of the dynamics has been pointed out also in \cite{ZhangY}.

We conclude that the results exposed in this section can be generalized to the case where the structure highlighted in Eq.\eqref{jump_sets} holds in another
arbitrary orthonormal basis $\{|\varphi_n\rangle\}_n$ on $\mathscr{H}$.

\subsection{\label{sec: lambda}Spectral bounds}

In this paragraph, we will provide non-trivial bounds on the quantities introduced in Eq.\eqref{eq_lambdas}, under the hypotheses of Theorem \ref{th: main_1}. 

For $\lambda_{\text{max}}$, one can directly utilize the Böttcher-Wenzel inequality (\cite{Bottcher2005}) $||[a,b]||_2\le \sqrt{2}||a||_2||b||_2$, to obtain
\begin{equation}
\label{bound_lambdamax}
\lambda_{\text{max}}\le 2\sum_{\alpha}\gamma_{\alpha}||L_{\alpha}||_2^2.
\end{equation}
Clearly, this upper bound holds in general, regardless of the specific form of the jumps and the condition of irreducibility. A more complicated task is to provide a non-trivial lower bound for $\lambda_{\text{min}}$. 
To achieve this, we will leverage the construction introduced in Definition \ref{def: jumpsgraph}, making use of known results in spectral theory of graphs.

\begin{lemma}
\label{lemma_spectralbounds}
Let $\mathcal{J}=\mathcal{J}_{\text{diag}}\cup\mathcal{J}_{\text{off}}$  be the set of jump operators in Eq.\eqref{jump_sets}, which we assume to be self-adjoint and irreducible,  where $\mathcal{J}_{\text{off}}=\{ L_{\alpha}=\xi_{\alpha}|n_{\alpha}\rangle\langle  m_{\alpha}|\}_{\alpha\neq\alpha_0}$ as in Eq.\eqref{jumpoff_struct}
and let $\{\gamma_{\alpha} \}_{\alpha\neq \alpha_0}$ be the associated dissipation rates. Then
\begin{equation}
\label{eq_boundslambdas}
\lambda_{\text{min}} \ge \min\Big\{1, \frac{8}{N(N-1)}\Big\}  \min\limits_{\alpha\neq\alpha_0}\{\gamma_{\alpha}|\xi_{\alpha}|^2 \}.
\end{equation}

\end{lemma}

\begin{proof}

Let us consider $\lambda_{\text{min}}$ defined in Eq.\eqref{eq_lambdas}. According  to the decomposition in Eq.\eqref{jump_sets}, we can split $\sum_{\alpha}=\sum_{\alpha_0}+\sum_{\alpha\neq\alpha_0}$, where the first contributions given in Eq.\eqref{eq_jumpdiag} belong to $\mathcal{J}_{\text{diag}}$.
Therefore, we have
\begin{equation}
\begin{aligned}
\lambda_{\text{min}}&\ge \inf\limits_{a\perp 1_H}\sum_{\alpha\neq\alpha_0}\gamma_{\alpha}\frac{||[L_{\alpha},a]||_2^2}{||a||_2^2}  \\
& =\underbrace{\inf\limits_{a \perp 1_H}\frac{(a,\tilde{\mathcal{K}}(a))_{HS}}{||a||_2^2}}_{=\tilde{\lambda}_{\text{min}}},
\end{aligned}
\end{equation}
where $\tilde{\mathcal{K}}(a)=\sum_{\alpha\neq\alpha_0}\gamma_{\alpha}[L_{\alpha}^{\dagger},[L_{\alpha},a]]$ has the same structure of the superoperator introduced in Eq.\eqref{eq_Kdef}. In particular, because of Lemma \ref{lemma_jumps}, the set $\{L_{\alpha}\}_{\alpha\neq \alpha_0}=\mathcal{J}_{\text{off}}$ is still self-adjoint and irreducible, hence $\tilde{\mathcal{K}}>0$ when restricted to the subspace orthogonal to the identity. We set up to provide a non-trivial lower bound to its smallest non-zero eigenvalue $\tilde{\lambda}_{\text{min}}$.

In order to do so, we express the superoperator $\tilde{\mathcal{K}}$ in the basis $\{|n\rangle\langle m|\}_{n,m=1,...,N}$ in Eq.\eqref{jump_sets}, obtaining the $N^2\times N^2$ matrix with entries 
\begin{equation}
\label{eq_Kmatrixelem_0}
\begin{aligned}
\tilde{\mathcal{K}}_{nm,pq}=& (|n\rangle\langle m|, \tilde{\mathcal{K}}(|p\rangle\langle q|))_{HS} \\
=&\sum_{\alpha\neq \alpha_0}\gamma_{\alpha}\Big[(L_{\alpha}^{\dagger}L_{\alpha})_{np}\delta_{mq}+(L_{\alpha}L_{\alpha}^{\dagger})_{qm}\delta_{np} \\
&- (L_{\alpha})_{np} (L_{\alpha}^{\dagger})_{qm}-(L_{\alpha})_{qm} (L_{\alpha}^{\dagger})_{np} \Big].
\end{aligned}
\end{equation}
Utilizing the expression $L_{\alpha}=\xi_{\alpha}|n_{\alpha}\rangle\langle m_{\alpha}|$, we get to
\begin{equation}
\label{eq_explicitK}
\begin{aligned}
\tilde{\mathcal{K}}_{nm,pq}=&\sum_{\alpha\neq \alpha_0}\gamma_{\alpha}
|\xi_{\alpha}|^2\Big(\delta_{n,m_{\alpha}}\delta_{np}\delta_{mq}+\delta_{m,n_{\alpha}}\delta_{np}\delta_{mq} \\
&- \delta_{n,n_{\alpha}}\delta_{m,n_{\alpha}}\delta_{p,m_{\alpha}}\delta_{q,m_{\alpha}} \\
&- \delta_{n,m_{\alpha}}\delta_{m,m_{\alpha}}\delta_{p,n_{\alpha}}\delta_{q,n_{\alpha}}\Big). \\
\end{aligned}
\end{equation}
From this expression and the fact that $n_{\alpha}\neq m_{\alpha}$, one can easily show that $\tilde{\mathcal{K}}_{nm,pp}=0$, $\forall n,m$ s.t. $n\neq m$, $\forall p$, which means that the superoperator has no mixed contributions between diagonal and off-diagonal basis elements. This is equivalent to say that
\begin{equation}
\begin{aligned}
\tilde{\mathcal{K}}(\Span\{|n\rangle\langle n|\}_n )\subseteq & \Span\{|n\rangle\langle n|\}_n , \\
\tilde{\mathcal{K}}(\Span\{|n\rangle\langle m|\}_{n\neq m} )\subseteq & \Span\{|n\rangle\langle m|\}_{n\neq m}.
\end{aligned}
\end{equation}

Therefore, we can address separately the spectrum of the superoperator restricted to the subspace of diagonal and off-diagonal basis elements:

\begin{enumerate}
\item \textbf{Diagonal subspace}. 
Setting $n=m,p=q$, we get
\begin{equation}
\begin{aligned}
\tilde{\mathcal{K}}_{nn,pp}=&\sum_{\alpha\neq\alpha_0}\gamma_{\alpha}|\xi_{\alpha}|^2\Big[\delta_{np}(\delta_{n,n_{\alpha}}+\delta_{n,m_{\alpha}} ) \\
&-\delta_{n,n_{\alpha}}\delta_{p,m_{\alpha}}-\delta_{n,m_{\alpha}}\delta_{p,n_{\alpha}} \Big].
\end{aligned}
\end{equation}
This expression can be interpreted as the Laplacian matrix $L_G=D_G-A_G$ of the weighted graph $G=(V,E,W)$ associated to the set $\{\sqrt{\gamma}_{\alpha}(L_{\alpha}+L_{\alpha}^{\dagger})\}_{\alpha\neq \alpha_0}$, with
nodes given by $V=\{|n\rangle\}_{n=1,...N}$, degree matrix
\begin{equation}
(D_G)_{np}=\delta_{np}\Big(\sum_{\alpha\neq \alpha_0}\gamma_{\alpha}(L_{\alpha}^{\dagger}L_{\alpha} + L_{\alpha}L_{\alpha}^{\dagger}) \Big)_{nn}
\end{equation}
and adjacency matrix 
\begin{equation}
(A_G)_{np}=\sum_{\alpha\neq\alpha_0}\gamma_{\alpha}\Big(|(L_{\alpha})_{np}|^2+ |(L_{\alpha})_{pn}|^2\Big).
\end{equation}
One can easily check that the important property
\begin{equation}
(D_G)_{nn}=\sum_{p}(A_G)_{np}
\end{equation}
is satisfied.
We observe that 
\begin{equation}
L_G\ge 2 \min\limits_{\alpha\neq \alpha_0}\{\gamma_{\alpha}|\xi_{\alpha}|^2 \} L_H,
\end{equation}
where $L_H$ is the Laplacian of the undirected simple graph $H=(V,E)$ associated to $\{\frac{L_{\alpha}}{|\xi_{\alpha}|} \}_{\alpha\neq \alpha_0}$. In particular, $\{\frac{L_{\alpha}}{|\xi_{\alpha}|} \}_{\alpha\neq \alpha_0}$ is clearly self-adjoint and irreducible, thus Lemma \ref{lemma_irrgraph} guarantees that $H$ is a connected graph. From known results in graph theory (see Appendix \ref{app: graph}), $L_H$ has a unique zero eigenvalue (corresponding to the identity operator for $\tilde{\mathcal{K}}$),  whereas the smallest non-zero eigenvalue can be bounded by the Mohar-McKay inequality in Eq.\eqref{eq_boundl2uni}. Putting all together, we obtain that the smallest non-zero eigenvalue of $L_G$ is bounded by
\begin{equation}
\label{eq_lambda2Kdiag}
\tilde{\mu}_2\ge \frac{8}{N(N-1)} \min\limits_{\alpha\neq \alpha_0}\{\gamma_{\alpha}|\xi_{\alpha}|^2 \},
\end{equation}
which provides a lower bound to the spectrum of $\tilde{\mathcal{K}}$ orthogonal to the identity, on the diagonal subspace.

\item \textbf{Off-diagonal subspace}. We will study now the remaining part of the spectrum of $\tilde{\mathcal{K}}$, which we know to be strictly positive. Considering Eq.\eqref{eq_explicitK} for  $p\neq q$, $n\neq m$, we obtain that 
\begin{equation}
 \hspace*{\leftmargini}
\begin{aligned}
\tilde{\mathcal{K}}_{nm,pq}
=\sum_{\alpha\neq \alpha_0}\gamma_{\alpha}
|\xi_{\alpha}|^2(\delta_{n,m_{\alpha}}+\delta_{m,n_{\alpha}} )\delta_{np}\delta_{mq}.
\end{aligned}
\end{equation}
Since the matrix is diagonal on such subspace, it means that all $N^2-N$ operators $\{|p\rangle\langle q| \}_{p\neq q}$ are actually eigenoperators with eigenvalue
\begin{equation}
\tilde{\mu}_{pq}=\sum_{\alpha\neq \alpha_0}\gamma_{\alpha}
|\xi_{\alpha}|^2(\delta_{p,m_{\alpha}}+\delta_{q,n_{\alpha}} )>0.
\end{equation}
In particular, we have that $\tilde{\mu}_{pq}\ge \min\limits_{\alpha\neq \alpha_0}\gamma_{\alpha}
|\xi_{\alpha}|^2$, $\forall p\neq q$. The latter expression and Eq.\eqref{eq_lambda2Kdiag} lead to the bound in Eq.\eqref{eq_boundslambdas}.

\end{enumerate}

\end{proof}

\subsection{\label{sec: example}Example: 3-level system}

Ultimately, as a concrete example we study a driven 3-level system. The qutrit is characterized by the energy levels $E_3>E_2>E_1$, and a coherent driving field between the ground state $|1\rangle$ and the excited state $|3\rangle$, namely
\begin{equation}
\label{eq_Hqutrit}
H_S(t)=\left(\begin{array}{ccc}
E_1 & 0 & h(t) \\
0 & E_2 & 0 \\
h(t) & 0 & E_3 \\
\end{array}
\right),
\end{equation}
where for now $h(t)$ is some analytic function such that $h(0)=0$. In the following,
it is convenient to remove from $H_S(t)$ the contribution $\frac{E_1+E_3}{2}$ proportional to the identity and work directly with the parameters  $\delta=\frac{E_3-E_1}{2}, \eta=E_2-\frac{E_1+E_3}{2}$, instead of $E_3,E_2,E_1$. We will assume for now that $\eta\neq 0$, leaving the discussion of the case $\eta=0$ to the end. The interaction with the environment in Eq.\eqref{eq_bosonicbath} is responsible for incoherent transitions between the levels $|3\rangle\rightarrow |1\rangle$ and $|3\rangle\rightarrow |2\rangle$, as given by
\begin{equation}
\label{eq_Aqutrit}
\begin{aligned}
A_1 &=|1\rangle\langle 3|+ |3\rangle\langle 1|, \\
A_2 &=|2\rangle\langle 3|+ |3\rangle\langle 2|. \\
\end{aligned}
\end{equation}
As derived in Appendix \ref{sec: app_mastereq}, the dynamics in the interaction picture is governed by the GKLS generator
\begin{equation}
\label{eq: gkls_qutrit}
\begin{aligned}
\mathcal{L}_{I,t}(\rho)=&-\imath[{H}_{I,LS},\rho]+\gamma_{0}\Big(L_{0}\rho L_{0}-\frac{1}{2}\{ L_{0}^2,\rho\} \Big)  \\
&\sum_{\nu=1,2,3} \Big[\gamma_{-,\nu}\Big(L_{\nu}\rho L_{\nu}^{\dagger}-\frac{1}{2}\{ L_{\nu}^{\dagger}L_{\nu},\rho\} \Big)  \\
&+\gamma_{+,\nu} \Big(L_{\nu}^{\dagger}\rho L_{\nu} -\frac{1}{2}\{ L_{\nu}L_{\nu}^{\dagger},\rho\} \Big) \Big],
\end{aligned}
\end{equation}
where 
\begin{equation}
\label{eq_jumps}
\begin{aligned}
L_0(t)=& \sin\theta(t) (|3\rangle\langle 3|-|1\rangle\langle 1|), \\
L_1(t)=& \cos\theta(t)|1\rangle\langle 3|, \\
L_2(t)=& \sin\frac{\theta(t)}{2}|1\rangle\langle 2|, \\
L_3(t)=& \cos\frac{\theta(t)}{2}|2\rangle\langle 3|, \\
\end{aligned}
\end{equation}
with $\tan\theta(t)=\frac{2h(t)}{E_3-E_1}\equiv \frac{h(t)}{\delta}$,
whereas the expressions for the rates are given in Eq.\eqref{eq_ratesqutrit}. We observe that these jumps are exactly in the form outlined in Eq.\eqref{jump_sets}.

As discussed in Section \ref{sec: omega}, the Hamiltonian term ${H}_{I,LS}(t)$ does not play any role in determining the quantities involved in Theorem \ref{th: main_2}, hence we can simply omit it in our analysis.

We are now ready to study the conditions for relaxation.

\begin{enumerate}
\item \textbf{Set of jump operators}. It is not difficult to show that the self-adjoint set of jump operators 
\begin{equation}
\mathcal{J}_{\text{off}}=\{ L_{\nu},L_{\nu}^{\dagger}\}_{\nu=1,2,3}
\end{equation}
is also irreducible. In particular, utilizing Definition \ref{def: jumpsgraph}, one can easily apply Lemma \ref{lemma_irrgraph} to the graph

\begin{center}
\begin{tikzpicture}[node distance=2cm]
    \node[circle, draw] (1) {$|1\rangle$};
    \node[circle, draw, right of=1] (2) {$|2\rangle$};
    \node[circle, draw, below of=1] (3) {$|3\rangle$};

    \draw (1) -- node[above] {$\omega_{12}$} (2);
    \draw (1) -- node[left] {$\omega_{13}$} (3);
    \draw (2) -- node[right] {$\omega_{23}$} (3);
\end{tikzpicture}
\end{center}

associated to such set, where  the weights $W(\{n,m\})=\omega_{nm}$ are given by
\begin{equation}
\label{eq_weightqutrit}
\begin{aligned}
\omega_{12}=&\omega_{21}=\sin^2\Big(\frac{\theta(t)}{2}\Big)=\frac{1}{2}\Big[1-\frac{\delta}{r(t)}\Big], \\
\omega_{23}=&\omega_{32}=\cos^2\Big(\frac{\theta(t)}{2}\Big)=\frac{1}{2}\Big[1+\frac{\delta}{r(t)}\Big], \\
\omega_{13}=&\omega_{31}=\cos^2\theta(t)=\frac{\delta^2}{r^2(t)}, \\
\end{aligned}
\end{equation}
where $r(t)=\sqrt{\delta^2+h^2(t)}$. We observe that $0<\delta\le r(t)$.

\item \textbf{Spectral bounds in Eq.\eqref{eq_lambdas}}.

The quantities in Eq.\eqref{eq_lambdas} can be bounded utilizing Eqs.\eqref{bound_lambdamax}, \eqref{eq_boundslambdas} and the expressions for the rates in Eq.\eqref{eq_ratesqutrit}. In this section, we will focus on $\lambda_{\text{min}}(t)$, which is crucial in determining the relaxation rate function $F(t)$. We have that
\begin{equation}
 \hspace*{\leftmargini}
\label{eq_lambdaqutrit}
\begin{aligned}
\lambda_{\text{min}}(t)\ge  2\pi \min\limits_{\omega\in \{2r(t),r(t)-\eta\} }\{J(\omega) \}\bar{n}(2 r(t) )\frac{\delta^2}{r^2(t)},
\end{aligned}
\end{equation}
where $J(\omega)$ is the bath spectral density, $\bar{n}(\omega)=(e^{\beta\omega}-1)^{-1}$.

The derivation of the latter equation can be found in Appendix \ref{sec: app_mastereq}.

\item \textbf{Instantaneous steady state}.
This is fully determined by the fact that the dissipation rates in Eq.\eqref{eq_ratesqutrit} obey the detailed balance condition in Eq.\eqref{eq_detailedbalance}, hence Lemma \ref{lemma_fixedth} 
guarantees that the unique instantaneous steady state is the thermal state $\omega_{\text{th}}(t)$ at inverse temperature $\beta$, with respect to the Hamiltonian in Eq.\eqref{eq_avgH}. In particular, its spectrum reads 
\begin{equation}
\label{eq_qutritomega}
\begin{aligned}
\omega_{-}(t)=&\frac{e^{\beta r(t)}}{e^{-\beta \eta}+ 2\sinh(\beta r(t))}, \\ 
\omega_{2}(t)=&\frac{e^{-\beta \eta}}{e^{-\beta \eta}+ 2\sinh(\beta r(t))}, \\ 
\omega_{+}(t)=&\frac{e^{-\beta r(t)}}{e^{-\beta \eta}+ 2\sinh(\beta r(t))}. \\ 
\end{aligned}
\end{equation}
Moreover, one can also derive the spectrum of $\dot{\omega}_{\text{th}}(t)$, which we conveniently normalize with $m(\omega_{\text{min}}(t))=\omega_+(t)$, namely
\begin{equation}
 \hspace*{\leftmargini}
\label{eq_qutritomegadot}
\begin{aligned}
\frac{\dot{\omega}_+(t)}{\omega_+(t)}&=-\beta\dot{r}(t)\Bigg[1+\frac{2\sinh(\beta r(t))}{e^{-\beta \eta}+ 2\sinh(\beta r(t))}\Bigg], \\
\frac{\dot{\omega}_-(t)}{\omega_+(t)}&=\beta\dot{r}(t)e^{2\beta r(t)}\Bigg[1-\frac{2\sinh(\beta r(t))}{e^{-\beta \eta}+ 2\sinh(\beta r(t))}\Bigg], \\
\frac{\dot{\omega}_2(t)}{\omega_+(t)}&=-\beta\dot{r}(t)e^{\beta(r(t)-\eta)}\frac{2\sinh(\beta r(t))}{e^{-\beta \eta}+ 2\sinh(\beta r(t))}. \\
\end{aligned}
\end{equation}

The smallest eigenvalue $m(\dot{\omega}_{th}(t))$ depends on the behavior of the function $\dot{r}(t)$ and reads
\begin{equation}
 \hspace*{\leftmargini}
\label{eq_qutritbounddot}
m(\dot{\omega}_{th}(t))=
\begin{cases}
\min\{\dot{\omega}_+(t),\dot{\omega}_2(t)\} & \text{if } \dot{r}(t)>0, \\
0  & \text{if } \dot{r}(t)=0, \\
\dot{\omega}_-(t) & \text{if } \dot{r}(t)<0. \\
\end{cases}
\end{equation}

Details on the derivation can be found in Appendix \ref{sec: app_mastereq}.

\end{enumerate}

Finally, we want to comment the case $\eta=0$. From the expressions in Eq.\eqref{eq_A3level}, one can easily see that the set of Bohr frequencies become degenerate and the jumps are now given by $\mathcal{J}^{(\eta=0)}=\{L_{\mu}^{(\eta=0)}(t), (L_{\mu}^{(\eta=0)})^{\dagger}(t)\}_{\mu=0,1,2}$, where
\begin{equation}
\begin{aligned}
L_{0}^{(\eta=0)}(t)&=L_0(t), \\
L_{1}^{(\eta=0)}(t)&=L_1(t), \\
L_{2}^{(\eta=0)}(t)&=L_2(t)-L_3(t).
\end{aligned}
\end{equation}
The results concerning the irreducibility of $\mathcal{J}^{(\eta=0)}$ and the instantaneous steady state are easily extended to this case. A specific mention, however, concerns the bound for $\lambda_{\text{min}}$, which cannot be derived anymore by applying straightforwardly Lemma \ref{lemma_spectralbounds}, since  $L_{2}^{(\eta=0)}(t)$ does not have the required structure in Eq.\eqref{jumpoff_struct}. Nevertheless, it is still possible to compute such bound via a direct diagonalization of the superoperator $\tilde{\mathcal{K}}$ in Eq.\eqref{eq_Kmatrixelem_0}. In this case, one can easily check that the diagonal and off-diagonal subspaces in the basis $\{|n\rangle\langle m|\}_{n,m=1,2,3}$ are still decoupled and that a simple lower bound is provided by
$\lambda_{\text{min}}^{(\eta=0)}(t)\ge \frac{1}{2} \min\limits_{\mu=1,2}\{ \gamma^{(\eta=0)}_{+,\mu}(t)\}$, where $\gamma^{(\eta=0)}_{+,\mu}(t)$ denotes the dissipation rates in Eq.\eqref{eq_ratesqutrit} in the limit $\eta=0$, namely
\begin{equation}
\lambda_{\text{min}}^{(\eta=0)}(t)\ge\pi \min\limits_{\omega\in \{2r(t),r(t)\} }\{J(\omega) \}\bar{n}(2 r(t) ).
\end{equation}

Putting all together, $F(t)$ in Eq.\eqref{eq_boundsF} can be bounded by
\begin{equation}
F(t)\ge \frac{1}{2}\Big[ e^{-2\beta r(t)}\lambda_{\text{min}}(t)+\frac{m(\dot{\omega}_{th}(t))}{m(\omega_{\text{min}}(t))}\Big],
\end{equation}
hence, if $\int\limits_{t_0}^{\infty}dt F(t)=\infty$, for some $t_0\ge 0$, the condition of relaxation in Eq.\eqref{eq_condrelax} is satisfied.

\section{\label{sec:concl}Conclusions}

In this paper, we have addressed the long-time dynamics of Markovian master equations. 
A general form of relaxation (weak relaxation) has been introduced. In contrast to the standard definition in the context of quantum dynamical semigroups (strong relaxation), our definition does not require the existence of a well-defined (stationary) state for $t\rightarrow \infty$, but rather a trajectory of states.

We have derived sufficient conditions that allow us to determine whether the dynamics described by time-dependent Markovian master equations is weakly relaxing. These conditions are a natural generalization of the ones provided by the Spohn-Frigerio Theorem \ref{th_spohnfrig} and also include an estimate of a relaxation rate function $F(t)$. Broadly speaking, $F(t)$ is expressed in terms of the spectrum of the instantaneous steady state, the dissipation rates and the norm of the jump operators. 
Focusing on a specific type of Markovian master equation for driven systems recently derived, we linked the latter quantities to the structure of the microscopic Hamiltonian for system and environment, introducing a graph theoretical structure associated to the set of jump operators. This allowed us to link the irreducibility of the set of jumps with the connectivity of the graph, as well as providing a bound for the quantities contributing to the relaxation rate function $F(t)$.
In this context, the specific instance of a driven 3-level system is analyzed.

Finally, as expressed in Corollary \ref{cor_nonMarkov}, our results are applicable even in the case of generic non-Markovian time-local master equations, which become asymptotically Markovian. This simple consequence opens up the possibility of formulating a more general theory of relaxation that extends beyond the Markovian regime.

In conclusion, we outline some possible extensions of the present work.

First, it is desirable to analyze the necessary conditions for relaxation, which eventually would provide a deeper understanding of the emergence and role of multiple asymptotic trajectories, dependent on the initial state preparation.

Second, the methods employed in our paper might be effectively extended to encompass more general cases of non-Markovian dynamics, which currently lie beyond the regime of applicability of our results.

\begin{acknowledgments}

Giovanni Di Meglio acknowledges fruitful discussions with Felix Ahnefeld. This work was supported by the ERC Synergy grant HyperQ (Grant no. 856432) and the BMBF under the funding program 'quantum technologies - from basic research to market' in the project CoGeQ (Grant no. 13N16101). D.C. was supported by the Polish National Science Centre project No. 2024/55/B/ST2/01781.

\end{acknowledgments}

\appendix

\section{\label{sec:app_Heisenberg}Heisenberg picture analysis}

In this appendix, we provide an alternative proof of the conditions for relaxation in the time-independent case, by approaching the problem in the Heisenberg picture. Given a GKLS generator $\mathcal{L}$, we recall that $\mathcal{L}^{\ddagger}$ denotes its adjoint with respect to the Hilbert-Schmidt scalar product.

Let $\omega$ be a positive operator, we can consider the following weighted scalar product
\begin{equation}
(a,b)_{\omega}=\Tr(\omega a^{\dagger}b),
\end{equation}
which defines the weighted norm $||a||_{\omega}$. Utilizing the spectrum of $\omega$, we can derive the useful inequality
\begin{equation}
\label{eq_normeq}
c_{\omega}||a||_2\le ||a||_{\omega}\le C_{\omega}||a||_2,
\end{equation}
where clearly $c_{\omega}=\sqrt{\inf \Sp(\omega)}$, $C_{\omega}=\sqrt{\sup \Sp(\omega)}$.

The connection between the convergence in the Heisenberg picture and the relaxation of the CPTP semigroup is provided by the following 
\begin{lemma}
\label{lemma_conv}
Let $\mathcal{L}$ be the generator of a CPTP semigroup with a faithful steady state $\omega$ and $a(t)=e^{t\mathcal{L}^{\ddagger}}(a)$, then if $\forall a\in\mathscr{B}$,
\begin{equation}
\lim\limits_{t\rightarrow\infty}a(t)=\Tr(\omega a)1_H,
\end{equation}
the semigroup is relaxing, i.e. $\lim\limits_{t\rightarrow\infty}e^{t\mathcal{L}}(\rho_0)=\omega$, $\forall\rho_0 \in\mathscr{S}_1^+$.
\end{lemma}
\begin{proof}
By hypothesis we have that
\begin{equation}
\Tr(a(t)\rho_0)\rightarrow \Tr(\omega a)\Tr(\rho_0)=\Tr(\omega a),
\end{equation}
$\forall a \in\mathscr{B}$, $\forall\rho \in\mathscr{S}_1^+$. But, this implies $\Tr(a(t)\rho_0)=\Tr(a \rho(t))$ and since this holds $\forall a$, the thesis follows.
\end{proof}

We are now ready to provide an alternative proof to the Spohn-Frigerio theorem (\cite{Frigerio_1, Frigerio_2}) stated in Theorem \ref{th_spohnfrig}

\begin{proof}

First, we want to show that any GKLS generator $\mathcal{L}$ that satisfies the conditions in Eq.\eqref{eq_condSF} has a faithful steady state.

Let $\mathcal{L}(\rho)=-\imath[H,\rho]+\sum_{\alpha}\gamma_{\alpha} (L_{\alpha}\rho L_{\alpha}^{\dagger}-\frac{1}{2}\{ L_{\alpha}^{\dagger} L_{\alpha},\rho\})$ be a GKLS generator as in the hypothesis, so we assume that for any $L_{\alpha}$, the adjoint $L_{\alpha}^{\dagger}$ appears in the generator.
Let $\omega$ be a steady state of the dynamics, which we know always exists, and let us assume that $\omega\ge 0$. We consider $|\psi\rangle\in\ker(\omega)$, then from $\mathcal{L}(\omega)=0$ we obtain 
\begin{equation}
\langle\psi|\mathcal{L}(\omega)|\psi\rangle=\sum_{\alpha}\gamma_{\alpha} \langle\psi|L_{\alpha}\omega L_{\alpha}^{\dagger}|\psi\rangle=0,
\end{equation}
which implies that $\sqrt{\omega}(L_{\alpha}^{\dagger}|\psi\rangle)=0$, $\forall\alpha$. Therefore, $L_{\alpha}^{\dagger}|\psi\rangle\in\ker(\omega)$, $\forall\alpha$, and since the set of jump operators includes the adjoint for each jump, we have that $L_{\alpha}|\psi\rangle\in\ker(\omega)$ too. Clearly, by the same argument it is not difficult to see that 
\begin{equation}
\begin{aligned}
\ker(\omega)= \Span\{ &|\psi\rangle, L_{\alpha}|\psi\rangle, L_{\alpha}^{\dagger}|\psi\rangle, L_{\beta} L_{\alpha}|\psi\rangle,  \\
&L_{\beta} L_{\alpha}^{\dagger}|\psi\rangle, ...\}_{\alpha,\beta,...},
\end{aligned}
\end{equation}
so that the image of the algebra generated by the identity and the jump operators is the kernel, namely $\alg\{1_H,L_{\alpha},L_{\alpha}^{\dagger}\}_{\alpha}(|\psi\rangle)=\ker(\omega)$. However, the von Neumann bicommutant theorem guarantees that the unital algebra $\alg\{1_H,L_{\alpha},L_{\alpha}^{\dagger}\}_{\alpha}$ is equal to its bicommutant, which is clearly the whole space of operators $\mathscr{B}$, due to the hypothesis of irreducibility. But this means that $\ker(\omega)=\mathscr{H}$ and thus $\omega=0$, which is a contradiction.

Now, let $\omega>0$ be such steady state, we consider on the space $\mathscr{B}$ the weighted scalar product induced by $\omega$. Hence, we can  decompose the space into two orthogonal subspaces with respect to $(a,b)_{\omega}$, as $\mathscr{B}=\mathscr{B}_0\oplus \mathscr{B}_1$, where
\begin{equation}
\mathscr{B}_0=\{ b| (1_H,b)_{\omega}=0\}
\end{equation}
contains elements orthogonal to the identity, i.e. satisfying $\Tr(b\omega)=0$. Clearly, since $\mathscr{B}_1$ is spanned by the identity operator, the space $\mathscr{B}_0$ is $N^2-1$-dimensional.

Due to Lemma \ref{lemma_conv}, the relaxation is guaranteed if
we prove that 
\begin{equation}
b(t)=a(t)-\Tr(\omega a)1_H\rightarrow 0,
\end{equation}
$\forall a\in\mathscr{B}$, where $a(t)=e^{t\mathcal{L}^{\ddagger}}(a)$.
By construction, $b(t)=e^{t\mathcal{L}^{\ddagger}}(a-\Tr(\omega a)1_H)$, with $b(0)\in \mathscr{B}_0$. Moreover, since $\mathscr{B}_0$ is an invariant subspace under $\Lambda_t^{\ddagger}$, we also have that $b(t)\in\mathscr{B}_0$, $\forall t\ge 0$.

Because of the identity
\begin{equation}
a^{\dagger}\mathcal{L}^{\ddagger}(a)+\mathcal{L}^{\ddagger}(a^{\dagger})a-\mathcal{L}^{\ddagger}(a^{\dagger}a)=-\sum_{\alpha}\gamma_{\alpha}[a^{\dagger},L_{\alpha}^{\dagger}][L_{\alpha},a]\le 0,
\end{equation}
one can derive the expression
\begin{equation}
\begin{aligned}
\label{eq_vara}
\ddt{||b(t)||_{\omega}^2}&=-\sum_{\alpha}\gamma_{\alpha}\Tr(\omega[b^{\dagger}(t),L_{\alpha}^{\dagger}][L_{\alpha},b(t)])\\
&=-\sum_{\alpha}\gamma_{\alpha}||[L_{\alpha},b(t)]||_{\omega}^2 ,
\end{aligned}
\end{equation}
hence making use of the norm equivalence in Eq.\eqref{eq_normeq}, we have that
\begin{equation}
\begin{aligned}
\ddt{||b(t)||_{\omega}^2} & \le -c_{\omega}^2\sum_{\alpha}\gamma_{\alpha}||[L_{\alpha},b(t)]||_{2}^2 \\
& =-c_{\omega}^2(b(t),\mathcal{K}(b(t)))_{HS},
\end{aligned}
\end{equation}
where $\mathcal{K}(a)=\sum_{\alpha}\gamma_{\alpha}[L_{\alpha}^{\dagger},[L_{\alpha},a]]$ is the superoperator introduced in Eq.\eqref{eq_firstK}. Therefore, by the same arguments following Eq.\eqref{eq_firstK}, $\mathcal{K}>0$ on the subspace $\mathscr{B}_0$, thus
\begin{equation}
\begin{aligned}
\ddt{||b(t)||_{\omega}^2}&\le -c_{\omega}^2\inf \Sp(\mathcal{K})||b(t)||_2^2 \\
&\le -\frac{c_{\omega}^2}{C_{\omega}^2}\inf \Sp(\mathcal{K})||b(t)||_{\omega}^2,
\end{aligned}
\end{equation}
and, invoking again Grönwall's lemma, we obtain the inequality
\begin{equation}
\label{eq_relax}
||b(t)||_{\omega}^2\le e^{-t\frac{c_{\omega}^2}{C_{\omega}^2}\inf \Sp(\mathcal{K}) } ||b(0)||_{\omega}^2\rightarrow 0,
\end{equation}
$ \forall b(0)\in \mathscr{B}_0 $, which proves the theorem.
\end{proof}

We underline that the upper bound obtained in Eq.\eqref{eq_relax} is the same as the time-independent case of Eq.\eqref{eq_boundsigma}.
In fact, we can go a bit further and prove that the description of the Heisenberg picture dynamics via $\Lambda_t^{\ddagger}$ in terms of the weighted norm is complementary to the description of the reduced map $\Lambda_{t}^{\natural}$ introduced in Eq.\eqref{eq_resc}, with respect to the Hilbert-Schmidt norm. In particular, we are going to prove a duality relation, expressed by the following result

\begin{lemma}
Let $\Lambda_t$ be a CPTP semigroup with a faithful steady state $\omega>0$, let $\Lambda_t^{\ddagger}$ be the adjoint and $\Lambda_t^{\natural}$ the reduced map in Eq.\eqref{eq_resc}, then $||\Lambda_t^{\natural}||_{2,2}=||\Lambda_t^{\ddagger}||_{\omega,\omega}$.
\end{lemma}
\begin{proof}
We start from the definition
\begin{equation}
||\Lambda_t^{\natural}||_{2,2}=\sup\limits_{a,b}\frac{|(a, \Lambda_t^{\natural}(b))_{HS}|}{||a||_2 ||b||_2},
\end{equation}
where
\begin{equation}
\begin{aligned}
|(a,\Lambda_t^{\natural}(b))_{HS}|& =|(\omega^{-1/2}a,\Lambda_t(\omega^{1/2}b) )_{HS}| \\
& =|(\Lambda^{\ddagger}_t(\omega^{-1/2}a),\omega^{1/2}b )_{HS}|,
\end{aligned}
\end{equation}
and, setting $\sigma=\omega^{-1/2}a$, we get
\begin{equation}
\begin{aligned}
||\Lambda_t^{\natural}||_{2,2}& =\sup\limits_{\sigma,b}\frac{|(\omega^{1/2}\Lambda^{\ddagger}_t(\sigma),b )_{HS}|}{||\omega^{1/2}\sigma||_2 ||b||_2} \\
&\le  \sup\limits_{\sigma,b} \frac{||\Lambda^{\ddagger}_t(\sigma)||_{\omega}||b||_2}{||\sigma||_{\omega}||b||_2}=||\Lambda^{\ddagger}_t||_{\omega,\omega}.
\end{aligned}
\end{equation}
Similarly, starting from $||\Lambda^{\ddagger}_t||_{\omega,\omega}$ the converse inequality can be proved.
\end{proof}

Unfortunately, the analysis of the time-dependent case in the Heisenberg picture becomes more problematic, since it would require the introduction of a time-dependent norm in terms of the instantaneous steady state $\omega(t)$.

\section{\label{app: graph} Graph theory}

In this appendix, we provide a self-contained introduction to the main notions and results of graph theory employed in this paper. A more detailed exposition can be found in any reference on spectral graph theory (see for instance \cite{Brouwer2011}).

\begin{definition}
\label{def: graph_1}
A directed simple graph is an ordered pair $G=(V,E)$, where
\begin{enumerate}
    \item $V\neq \emptyset$ is a finite set of vertices (or nodes);
    \item $E\subseteq \{ (v,w)| v,w\in V, v\neq w \}$ is the set of direct edges (or arcs).
\end{enumerate}
\end{definition}

\begin{definition}
\label{def: graph_2}
An undirected simple graph is an ordered pair $G=(V,E)$, where
\begin{enumerate}
    \item $V\neq \emptyset$ is a finite set of vertices (or nodes);
    \item $E\subseteq \{ \{v,w\}| v,w\in V, v\neq w \}$ is the set of edges (or links).
\end{enumerate}
\end{definition}

The cardinalities $|V|$ and $|E|$ define, respectively, the order and the size of $G$.

We notice that the edges in a direct graph are defined as the ordered cartesian product $(v,w)\in V\times V$, whereas for the undirected graph we have the unordered pair $\{v,w\}$. 

Moreover, by definition a simple graph is characterized by the absence of multiple edges connecting the same pair of vertices, as well as loops. Multigraphs allow to encompass these additional features. 

In addition, one can also associate different numbers to each edge, introducing the notion of weighted graph:
\begin{definition}
A weighted undirected simple graph is an ordered pair $G=(V,E,W)$, where
\begin{enumerate}
\item $V\neq \emptyset$ is a finite set of vertices (or nodes);
\item $E\subseteq \{ \{v,w\}| v,w\in V, v\neq w \}$ is the set of edges (or links);
\item $W: E\rightarrow \mathbb{R}$ is a function that assigns a real number to each edge. 
\end{enumerate}
\end{definition}

As clarifying examples, the representation of a directed simple graph with
\begin{enumerate}
    \item $V=\{a,b,c,d\}$,
    \item $E=\{ (a,b),(a,c),(d,b),(d,c) \}$,
\end{enumerate}

is

\begin{center}
\begin{tikzpicture}[nodes={circle, draw}]
  \graph {
    a -> {b, c} <- d
  };
\end{tikzpicture}
\end{center}

An undirected simple graph with 
\begin{enumerate}
    \item $V=\{1,2,3,4\}$,
    \item $E=\{ \{1,3\}, \{1,2\} , \{1,4\} ,\{2,4\}  \}$,
\end{enumerate}

is instead represented by

\begin{center}
\begin{tikzpicture}[nodes={circle, draw}]
\graph{3 -- 1 -- 4,
       1--2--4 };
\end{tikzpicture}
\end{center}

An instance of directed multigraph is 

\begin{center}
\begin{tikzpicture}[nodes={circle, draw}]
  \graph  {
    A ->[bend left] B,
    A ->[bend right] B,
    B ->[bend left] A,

    B ->[bend left] C,

    A ->[bend left] C,
    C ->[bend left] A,

    B ->[loop right] B,
  };
\end{tikzpicture}
\end{center}

Finally, a weighted graph with
\begin{enumerate}
    \item $V=\{v_1,v_2,v_3\}$,
    \item $E=\{ \{v_1,v_2\}, \{v_1,v_3\} , \{v_2,v_3\}   \}$,
    \item $W(\{v_n,v_m\})=\omega_{nm}=\omega_{mn}$,
\end{enumerate}
is represented as 

\begin{center}
\begin{tikzpicture}[node distance=2cm]
    \node[circle, draw] (1) {$v_1$};
    \node[circle, draw, right of=1] (2) {$v_2$};
    \node[circle, draw, below of=1] (3) {$v_3$};

    \draw (1) -- node[above] {$\omega_{12}$} (2);
    \draw (1) -- node[left] {$\omega_{13}$} (3);
    \draw (2) -- node[right] {$\omega_{23}$} (3);
\end{tikzpicture}
\end{center}

In our analysis, we will focus uniquely on undirected simple graphs, possibly with positive weights. Moreover, for the sake of notation, the set of vertices $V=\{v_1,v_2,...,v_N \}$ will be equally identified with the set $\{1,...,N\}$ and similarly for the edges $E= \{ \{n,m\}| v_n,v_m\in V, v_n\neq v_m \}$,

Now we can introduce some useful matrices associated to a graph.

\begin{definition}
Let $G=(V,E,W)$ be a weighted undirected simple graph, then the degree $\Deg(v)$ of a vertex $v\in V$ is the sum of all weights of the edges incident to it, namely
\begin{equation}
\Deg(v)=\sum_{w\text{ s.t. } \{v,w\}\in E} W(\{v,w\}).
\end{equation}
In the case of unweighted graph, this is exactly the number of edges incident to it. The $|V|\times |V|$ diagonal matrix $D_G$ with entries $(D_G)_{nm}=\delta_{nm}\Deg(v_n)$ is called degree matrix of $G$.
\end{definition}

In an unweighted simple graph, the absence of multiple connections between the same two vertices implies that $\Deg(v)\le |V|-1$.

\begin{definition}
Let $G=(V,E,W)$ be a weighted undirected simple graph with $|V|=N$, the $N\times N$ matrix $A_G$ with entries
\begin{equation}
(A_G)_{nm}=\begin{cases}
W(\{v_n,v_m\}) & \text{if  } \{v_n,v_m\}\in E \\
0 & \text{otherwise}
\end{cases}
\end{equation}
is called adjacency matrix of $G$.
\end{definition}
By construction, for an undirected graph $A_G$ is symmetric. 
We underline that, both matrices $A_G$, $D_G$ are not uniquely defined, as we can always reorder the elements in $V$ without changing the graph. Clearly, the properties of a graph do not depend on such freedom of choice.

A crucial quantity in graph theory is the following
\begin{definition}
Let $G=(V,E,W)$ be a weighted undirected simple graph, the matrix $L_G=D_G-A_G$ is called Laplacian matrix of $G$.
\end{definition}

The following properties generally hold:
\begin{enumerate}
\item $L_G$ is symmetric.
\item If $W(e)\ge 0$, $\forall e\in E$ (positive weights) then $L_G$ is positive semi-definite.
To show this, we first notice that $L_G$ can be decomposed as $L_G=\sum_{ \{n,m\}\in E} M_{nm}$, where 
\begin{equation}
\label{eq_laplaciansum}
(M_{nm})_{pq}=
\begin{cases}
W_{nm} & \text{if  $(p,q)=(n,n)$ or $(p,q)=(m,m)$  }\\
-W_{nm} & \text{if  $(p,q)=(n,m)$ or $(p,q)=(m,n)$  }\\
0 & \text{otherwise}
\end{cases}.
\end{equation}
Geometrically, this corresponds to observe that the graph can be seen as sum of smaller graphs $G_{nm}=(V_{nm},E_{nm},W_{nm})$ with $V=\{v_n,v_m\}$ and only one weighted edge connecting the two vertices, whose Laplacian is exactly $M_{nm}$. Now, let us consider $\underline{x}\in \mathbb{R}^N$ an arbitrary vector, where $N$ is the order of $G$, then 
\begin{equation}
\label{eq_quadraticLG}
\begin{aligned}
\underline{x}^T L_G \underline{x} & = \sum_{ \{n,m\}\in E} \underline{x}^T  M_{nm}\underline{x}\\
&=\sum_{ \{n,m\}\in E}\underbrace{\omega_{nm}}_{\ge 0} (x_n-x_m)^2\ge 0,
\end{aligned}
\end{equation}
which shows that the quadratic form is positive semi-definite.

\item By construction, we have that
\begin{equation}
\sum_m (L_G)_{nm}=\sum_n (L_G)_{nm}=0.
\end{equation}

\item $L_G \underline{1}=0$, where $\underline{1}=(1,...,1)^T$. This simply follows from Eq.\eqref{eq_quadraticLG}.
\end{enumerate}

Other properties are more specific of certain classes of graphs.

\begin{definition}
Let $G=(V,E)$ be a simple graph, then two vertices $v,w\in V$ are connected if $G$ contains a sequence of edges (i.e. path) with $v,w$ as endpoints.  If every pair $\{v,w\}$ of different vertices is connected, then the graph is said to be connected.
\end{definition}

The connectivity of a graph is linked to the null space of its Laplacian via this important result
\begin{lemma}
Let $G=(V,E)$ be an undirected simple graph, then $G$ is connected if and only if $\dim\ker (L_G)=1$.
\end{lemma}

\begin{proof}
Let $\lambda_1=0$ be the zero eigenvalue of $L_G$. First, we show that if $\lambda_1$ has multiplicity $1$, then $G$ must be connected. To this end, let us assume that $G$ is disconnected, then upon reordering of the vertices the Laplacian takes the block-diagonal form
\begin{equation}
L_G=\left[
\begin{array}{cc}
L_1 & 0 \\
0 & L_2 \\
\end{array}
\right],
\end{equation}
where $L_1,L_2$ are the Laplacian of the disconnected graphs of $G$. Thus,
\begin{equation}
\begin{array}{ccc}
\left[
\begin{array}{c}
\underline{1} \\
0 \\
\end{array}
\right], &
\left[
\begin{array}{c}
0 \\
\underline{1} \\
\end{array}
\right], 
\end{array}
\end{equation}
are both eigenvectors of $L_G$ with zero eigenvalue, which contradicts the hypothesis.

Now, let us assume that $G$ is connected, then any $\underline{x}=(x_1,...,x_N)^T\in \ker(L_G)$ must satisfy $\underline{x}^T L_G \underline{x}=0$ and because of Eq.\eqref{eq_quadraticLG} we have that $x_n=x_m$ $\forall \{n,m\} \in E$, which implies that $x_n=x_m$ $\forall n,m$ since $G$ is connected, i.e. $\underline{x}\in \Span\{ \underline{1}\}$.

\end{proof}

Beside $\lambda_1=0$, the smallest non-zero eigenvalue $\lambda_2$ and the largest eigenvalue $\lambda_{\infty}$ play a crucial role in the characterization of the graphs. In the following, we will focus uniquely on $\lambda_2$, which is also called "algebraic connectivity" or Fiedler value \cite{Fiedler1973}.

\begin{definition}
Let $G=(V,E)$ be an undirected graph and let $v,w\in V$ be two connected vertices, then the distance $d(v,w)$ between them is the number of edges in the shortest path $P_{vw}$ with endpoints $v,w$.
\end{definition}

Notice that, the shortest path needs not to be unique.

\begin{definition}
Let $G=(V,E)$ be a connected undirected graph, then the diameter $\Diam(G)$ is the largest distance between any pair of its vertices.
\end{definition}

Conventionally, if $G$ is disconnected then $\Diam(G)=\infty$. In particular, we are going to explore now a useful relation between diameter of a graph and its Fiedler value(\cite{Mohar1991a,Mohar1991b}). First, we need

\begin{lemma}
\label{lemma_ce}
Let $G=(V,E)$ be a connected undirected simple graph with $|V|=N$. For any edge $e\in E$, let $c_e$ be the number of pairs $\{v,w\}$ of vertices whose shortest path $P_{vw}$ contains $e$ as edge, then $c_e\le \frac{N^2}{4}$.
\end{lemma}

\begin{proof}
Let us consider a generic edge $e\in E$ with endpoints the vertices $x,y$. If there exists a third vertex $z$ such that $x,y,z$ are adjacent and thus form a triangle, then the pairs  $\{x,z\}$ and  $\{y,z\}$
have shortest paths $P_{xz}$ and $P_{yz}$ that do not include $e$ as edge. Hence, the extreme situation which provides the largest number of paths containing $e$ is when cutting the edge separates the graph $G$ into two disconnected graphs $G_1, G_2$ with number of vertices $|V_1|, |V_2|=N-|V_1|$, where $|V_1|=1,..,N-1$ (since at least one vertex must be in both graphs). In this case,  any pair $\{v_1,v_2\}$ with $v_1\in V_1, v_2\in V_2$ has shortest path $P_{12}$ that contains $e$. The number of possible pairings of this sort is $|V_1||V_2|$, thus
\begin{equation}
c_e\le \max\limits_{|V_1|=1,...,N-1} |V_1| (N-|V_1|)\le \frac{N^2}{4}.
\end{equation}

\end{proof}

Now we can prove

\begin{theorem}[\textbf{Mohar-McKay bound}]
\label{th_moharbound}
Let $G=(V,E)$ be a connected undirected simple graph, with $|V|=N$, and let $\lambda_2$ be its Fiedler value, then 
\begin{equation}
\lambda_2\ge \frac{4}{N \Diam(G)}.
\end{equation}
\end{theorem}

\begin{proof}
Let $L_G$ be the graph Laplacian and let $\underline{x}\in \mathbb{R}^N$ be a generic vector, then for a connected graph we have that
\begin{equation}
\lambda_2=\inf\limits_{\underline{x}\perp \underline{1}}\frac{\underline{x}^T L_G \underline{x}}{||\underline{x}||^2}.
\end{equation}
In particular, making use of Eq. \eqref{eq_quadraticLG}, we get
\begin{equation}
\label{eq_lambda2}
\lambda_2=\inf\limits_{\underline{x}\perp \underline{1}}  \frac{\sum_{\{n,m\}\in E} (x_n-x_m)^2}{\sum_n^N x_n^2}\equiv \inf\limits_{\underline{x}\perp \underline{1}}  \frac{\sum_{e\in E}\Delta_e^2}{\sum_n^N x_n^2},
\end{equation}
where we introduced the notation $\Delta_e=(x_n-x_m)$ for an edge $e$ with endpoints the vertices $\{n,m\}$.

Now, let us consider the sum over all possible pairings (not only the ones in $E$), namely
\begin{equation}
\label{eq_sumtotpairs}
\begin{aligned}
\sum_{\{n,m\}} (x_n-x_m)^2& =\frac{1}{2}\sum_n^N\sum_{m\neq n}^{N-1}(x_n^2+x_m^2-2x_n x_m) \\
&= \frac{1}{2}\Big( \sum_n^N 2(N-1)x_n^2+2\sum_n^N x_n^2 \Big) \\
&= N\sum_n^N x_n^2 ,
\end{aligned}
\end{equation}
where we utilized the fact that $\sum_{m\neq n}^{N-1}x_m=-x_n$, which is nothing but the condition of orthogonality $\underline{x}\perp \underline{1}$.

As next step, we consider a pair of vertices $\{v_n,v_m\}$ and let $P_{nm}$ be the shortest path that connects them, with length $l_{nm}$ corresponding to the number of edges in the path; clearly, $l_{nm}\le \Diam(G)$. If we consider the vertices $w_1,w_2...$ inside the path $P_{nm}$, then

\begin{equation}
x_n-x_m=(x_n-x_{w_1})+ (x_{w_1}-x_{w_2})+ ...=\sum_{e\in P_{nm}}\Delta_e,
\end{equation}
where the sum runs over the edges of $P_{nm}$. Making use of the general inequality $(\sum_i^N a_i)^2\le N\sum_i^N a_i^2$, we have that
\begin{equation}
\begin{aligned}
(x_n-x_m)^2 =\Big( \sum_{e\in P_{nm}}\Delta_e \Big)^2 & \le l_{nm}\sum_{e\in P_{nm}}\Delta_e^2  \\
&\le \Diam(G)\sum_{e\in P_{nm}}\Delta_e^2,
\end{aligned}
\end{equation}
therefore summing over all possible pairings leads to
\begin{equation}
\label{eq_sumpairs}
\begin{aligned}
\sum_{\{n,m\}} (x_n-x_m)^2&\le \Diam(G)\sum_{\{n,m\}}\sum_{e\in P_{nm}}\Delta_e^2 \\
& =  \Diam(G)\sum_{e\in E}\Delta_e^2 c_e,
\end{aligned}
\end{equation}
with $c_e$ being the number of paths $P_{nm}$ which contain $e$ as edge. Finally, combining the bound in Lemma \ref{lemma_ce} with Eqs. \eqref{eq_lambda2}, \eqref{eq_sumtotpairs} and \eqref{eq_sumpairs}, we arrive at
\begin{equation}
N\le \Diam(G)\frac{N^2}{4}\lambda_2.
\end{equation}

\end{proof}

Since in a connected simple graph the largest diameter that can be attained is $N-1$ (in the case of a path graph), then we obtain the universal bound
\begin{equation}
\label{eq_boundl2uni}
\lambda_2\ge \frac{4}{N(N-1)}.
\end{equation}

We underline that many other inequalities of this sort are discussed in the literature; however, we will make use of the one in Eq. \eqref{eq_boundl2uni}, due to its simplicity and usefulness for our purposes.

\section{\label{sec: app_mastereq}Master equation driven 3-level system}

In this appendix, we derive microscopically the generator in Eq.\eqref{eq: gkls_qutrit}.
We consider the system Hamiltonian in Eq.\eqref{eq_Hqutrit} and the bosonic environment in Eq.\eqref{eq_bosonicbath}, with $A_1,A_2$ in the interaction Hamiltonian given in Eq.\eqref{eq_Aqutrit} and $\rho_E$ thermal state at inverse temperature $\beta$. As explained in the main text, we introduce the useful parameters $\delta=\frac{E_3-E_1}{2}, \eta=E_2-\frac{E_1+E_3}{2}$ by removing the identity $\frac{E_1+E_3}{2}$ from $H_S(t)$. In particular, $-\delta<\eta<\delta$. In addition, we are going to use the function $r(t)=\sqrt{\delta^2+h^2(t)}$ instead of $h(t)$.

Let us start from the eigendecomposition of the system Hamiltonian, which is given by 
\begin{equation}
H_S(t)=\sum_{n=\pm,2}\epsilon_n(t)|\epsilon_n(t)\rangle\langle \epsilon_n(t)|,
\end{equation}
where 
\begin{equation}
\label{eq_spectrqutrit}
\begin{aligned}
\epsilon_{\pm}(t)=& \pm r(t),  \\
\epsilon_{2}(t)=& \eta,
\end{aligned}
\end{equation}
and
\begin{equation}
\begin{aligned}
|\epsilon_+(t)\rangle=&\Big( \sin\frac{\theta(t)}{2}, 0, \cos\frac{\theta(t)}{2} \Big)^T,  \\
|\epsilon_-(t)\rangle=&\Big( \cos\frac{\theta(t)}{2}, 0, -\sin\frac{\theta(t)}{2} \Big)^T,  \\
|\epsilon_2(t)\rangle=& (0,1,0)^T, \\
\end{aligned}
\end{equation}
with $\cos\theta(t)= \frac{\delta}{r(t)}$. We notice that, for any $h(t)$ we have $\epsilon_+(t)\ge \delta$ and $\epsilon_-(t)\le -\delta$, which implies that
$\epsilon_+(t)>\epsilon_2(t)>\epsilon_-(t)$.

Making use of Eq.\eqref{eq_jumpandhls}, the time-evolution operator
\begin{equation}
\begin{aligned}
V(t)=\sum_{n=\pm,\epsilon} \exp\Big(-\imath\int\limits_0^t ds \epsilon_n(s)\Big) |\epsilon_n(t)\rangle\langle \epsilon_n(0)|
\end{aligned}
\end{equation}
allows us to compute
\begin{equation}
\label{eq_A3level}
\begin{aligned}
V^{\dagger}(t)A_1 V(t) =&  \sin\theta(t) (|3\rangle\langle 3|-|1\rangle\langle 1|) \\
&+\exp\Big(\imath\int\limits_0^t ds \Omega_{-+}(t) \Big)\cos\theta(t)|1\rangle\langle 3| \\
&+\exp\Big(-\imath\int\limits_0^t ds \Omega_{-+}(t) \Big)\cos\theta(t)|3\rangle\langle 1|, \\
V^{\dagger}(t)A_2 V(t) =& -\exp\Big(\imath\int\limits_0^t ds \Omega_{-2}(t) \Big)\sin\frac{\theta(t)}{2} |1\rangle\langle 2| \\
&-\exp\Big(-\imath\int\limits_0^t ds \Omega_{-2}(t) \Big)\sin\frac{\theta(t)}{2} |2\rangle\langle 1| \\
&+\exp\Big(\imath\int\limits_0^t ds \Omega_{2+}(t) \Big)\cos\frac{\theta(t)}{2} |2\rangle\langle 3| \\
&+\exp\Big(-\imath\int\limits_0^t ds \Omega_{2+}(t) \Big)\cos\frac{\theta(t)}{2} |3\rangle\langle 2|,\\
\end{aligned}
\end{equation}
where $\Omega_{nm}(t)=\epsilon_n(t)-\epsilon_m(t)$ for $n,m=\pm,2$.
This leads to the jumps in Eq.\eqref{eq_jumps}, corresponding to the Bohr frequencies 
\begin{equation}
\label{eq_bohrfreqqutrit}
\begin{aligned}
\Omega_0(t)=&0, \\
\Omega_1(t)=& \epsilon_+(t)-\epsilon_-(t)=2r(t), \\
\Omega_2(t)=& \epsilon_2(t)-\epsilon_-(t)=\eta+r(t), \\
\Omega_3(t)=& \epsilon_+(t)-\epsilon_2(t)=r(t)-\eta. \\
\end{aligned}
\end{equation}

The set of environment correlation functions extracted from Eq.\eqref{eq_bosonicbath} is 
\begin{equation}
\begin{aligned}
R_{i j}(x)=&\Tr_E\Big(e^{\imath H_{E}x}B_{i}e^{-\imath H_{E}x}B_{j}\rho_{E}\Big) \\
=&\delta_{ij}\int\limits_0^{\infty}dw J_{j}(w)\Big((1+\bar{n}(w))e^{-\imath w x}  + \bar{n}(w)e^{\imath w x}\Big),
\end{aligned}
\end{equation}
where $J_{j}(w)=\int\limits_0^{\infty} dk \delta(w-\omega_{j}(k))g_{j}^2(k)$ are the bath spectral densities for $j=1,2$ and $\bar{n}(w)=(e^{\beta w}-1)^{-1}$ is the thermal mean occupation number.
Hence, from Eq.\eqref{eq_dissratesdef} we obtain  
\begin{equation}
\begin{aligned}
\int\limits_{-\infty}^{\infty} dx R_{jj}(x)e^{\mp \imath x \omega } =\begin{cases}
2\pi J_{j}(\omega)(1+\bar{n}(\omega)) & \text{if $\omega >0$} \\
2\pi J_{j}(|\omega|)\bar{n}(|\omega| )& \text{if $\omega <0$}
\end{cases},
\end{aligned}
\end{equation}
from which we can derive the dissipation rates 
\begin{equation}
\label{eq_ratesqutrit}
\begin{aligned}
\gamma_0(t)=&\lim\limits_{\omega\rightarrow 0}2\pi J_1(\omega)\bar{n}(\omega), \\
\gamma_{+,\nu}(t)=&2\pi J_{\nu}(\Omega_\nu(t))\bar{n}(\Omega_\nu(t)) , \\
\gamma_{-,\nu}(t)=&2\pi J_{\nu}(\Omega_\nu(t))\Big[1+\bar{n}(\Omega_\nu(t))\Big], \\
\end{aligned}
\end{equation}
for $\nu=1,2,3$ and  $J_{3}(\omega)\equiv J_2(\omega)$. For simplicity, we can assume that both interactions mediated by $A_1, A_2$ are described by the same spectral density.

In the following, we will provide details on the derivation of several quantities in the main text.

\begin{enumerate}
\item \textbf{Spectral bound in Eq.\eqref{bound_lambdamax}}. For $\lambda_{\text{max}}(t)$, one simply obtains 
\begin{equation}
 \hspace*{\leftmargini}
\begin{aligned}
\lambda_{\text{max}}(t) \le& 2\Big[ \gamma_0(t)\sin^2\theta(t) + (\gamma_{+,1}+\gamma_{-,1})\cos^2\theta(t) \\
&+ (\gamma_{+,2}++\gamma_{-,2})\sin^2\frac{\theta(t)}{2}\\
&+(\gamma_{+,3}+\gamma_{-,3})\cos^2\frac{\theta(t)}{2}\Big].  \\
\end{aligned}
\end{equation}

\item \textbf{Spectral bound in Eq.\eqref{eq_boundslambdas}}. 
We observe that the weight $\omega_{12}(t)$ in Eq.\eqref{eq_weightqutrit}  can be zero at certain times, when $r(t)=\delta$. Therefore, we provide a lower bound on $\lambda_{\text{min}}(t)$ working with the subgraph with edge set $E=\{ \{1,3\}, \{2,3\}\}$, which is clearly connected at any time and with strictly positive weights, thus
\begin{equation}
\hspace*{\leftmargini}
\lambda_{\text{min}}(t)  \ge  \min\limits_{\nu=1,3}\{ \gamma_{+,\nu}||L_{\nu}(t)||_2^2, \gamma_{-,\nu}||L_{\nu}(t)||_2^2\}.
\end{equation}
Making use of the fact that $\gamma_{-,\nu}(t)\ge \gamma_{+,\nu}(t)$, for $\forall \nu$, we get
\begin{equation}
\hspace*{\leftmargini}
\lambda_{\text{min}}(t)  \ge  \min\limits_{\nu=1,3}\{ \gamma_{+,\nu}(t)\}\underbrace{\min\Big\{\cos^2{\theta}(t), \cos^2\frac{\theta(t)}{2}  \Big\} }_{=g(t)}.
\end{equation}

In order to compute $\min\limits_{\nu=1,3}\{ \gamma_{+,\nu}(t)\}$, one can observe that $\bar{n} \Omega_{\nu}\ge \bar{n}\big(\max\limits_{\nu=1,3}(\Omega_{\nu})\big)$, where $\max\limits_{\nu=1,3}(\Omega_{\nu})=\Omega_1(t)$. 
Moreover, plugging $\cos\theta(t)=\frac{\delta}{r(t)}$ in the expression for $g(t)$, one arrives at Eq.\eqref{eq_lambdaqutrit}.

\item \textbf{Instantaneous steady state}. 
Since the dissipation rates in Eq.\eqref{eq_ratesqutrit} obey the detailed balance condition in Eq.\eqref{eq_detailedbalance}, Lemma \ref{lemma_fixedth} guarantees that the unique instantaneous steady state is the thermal state $\omega_{\text{th}}(t)$, with respect to the Hamiltonian in Eq.\eqref{eq_avgH}, which reads in this case
\begin{equation}
H_{\text{avg}}=\left(
\begin{array}{ccc}
\epsilon_-(t) & 0  & 0 \\
 0 & \epsilon_2(t)  & 0 \\
0   & 0 &  \epsilon_+(t) \\
\end{array}
\right).
\end{equation}

Therefore, indicating with $\omega_n(t)=\frac{1}{Z(t)}e^{-\beta \epsilon_n(t)}$ the spectrum,
where $Z(t)=\sum_{n=\pm,2}e^{-\beta\epsilon_n(t)}$, we obtain Eq.\eqref{eq_qutritomega}.

The spectrum of $\dot{\omega}_{\text{th}}(t)$ is simply given by the derivative of $\omega_n(t)$. Therefore, utilizing the relation
\begin{equation}
\dot Z(t)=-\beta \dot{r}(t) (\omega_+(t)-\omega_-(t)) Z(t),
\end{equation}
one obtains
\begin{equation}
 \hspace*{\leftmargini}
\label{eq_dotomega}
\begin{aligned}
\dot{\omega}_+(t)=& - \beta \dot{r}(t) \omega_+(t)\big[  (\omega_-(t)- \omega_+(t))+1 \big], \\
\dot{\omega}_-(t)=&  \beta \dot{r}(t) \omega_-(t)\big[1-  (\omega_-(t)- \omega_+(t)) \big], \\
\dot{\omega}_2(t)=& - \beta \dot{r}(t) \omega_2(t)(\omega_-(t)- \omega_+(t)).\\
\end{aligned}
\end{equation}
which leads to Eq.\eqref{eq_qutritomegadot}. 

Finally, we want to compute now the smallest eigenvalue of $\dot{\omega}_{\text{th}}(t)$. Taking into account that $0<\omega_-(t)-\omega_+(t)<1$, we have the cases
\begin{equation}
 \hspace*{\leftmargini}
\begin{array}{cccc}
\dot{r}(t)>0 \Rightarrow & \dot{\omega}_+(t)<0, & \dot{\omega}_-(t)>0, & \dot{\omega}_2(t) <0, \\
\dot{r}(t)<0 \Rightarrow  & \dot{\omega}_+(t)>0, & \dot{\omega}_-(t)<0, & \dot{\omega}_2(t) >0, \\
\end{array}
\end{equation}
which provide Eq.\eqref{eq_qutritbounddot}.

\end{enumerate}

\bibliography{refs} 

\end{document}